\documentclass[a4paper, 12pt]{article}
\usepackage[T1]{fontenc}
\usepackage[utf8]{inputenc}
\usepackage[english]{babel}
\usepackage[a4paper, margin=2.54cm]{geometry}
\usepackage{needspace}
\usepackage{libertine}
\usepackage{inconsolata}
\usepackage{amsmath,amsthm,amssymb}
\usepackage{amscd}
\usepackage{enumerate}
\usepackage{enumitem}
\usepackage{url}
\usepackage{stmaryrd}
\usepackage{graphicx}
\usepackage{xspace}
\usepackage{comment}
\usepackage{float}
\usepackage{cancel}
\usepackage{authblk}
\usepackage{esvect}
\usepackage{thmtools}
\usepackage{thm-restate}

\usepackage[margin=1.5cm]{caption}
\usepackage{xcolor}
\definecolor{refcolor}{RGB}{150,0,0}
\usepackage[colorlinks=true, allcolors=refcolor]{hyperref}

\usepackage{tikz,pgf}
\usetikzlibrary{positioning}
\usetikzlibrary{shapes.geometric}
\usetikzlibrary{arrows}
\usetikzlibrary{arrows.meta}
\usetikzlibrary{bbox}

\declaretheorem[name=Theorem, numberwithin=section]{theorem}
\declaretheorem[name=Theorem,unnumbered=true]{theorem*}
\declaretheorem[name=Lemma, sibling=theorem]{lemma}
\declaretheorem[name=Proposition, sibling=theorem]{proposition}
\declaretheorem[name=Definition, sibling=theorem]{definition}
\declaretheorem[name=Corollary, sibling=theorem]{corollary}
\declaretheorem[name=Question, sibling=theorem]{question}

\newcommand{\cuts}{\textup{\textsc{DisjointSeparators}}}
\newcommand{\sat}{\textup{\textsc{$3$-SAT}}}
\newcommand{\planarsat}{\textup{\textsc{Planar $3$-SAT}}}

\newcommand{\generalcuts}{\textup{\textsc{GeneralDisjointSeparators}}}

\newcommand{\NP}{{\sf NP}}
\newcommand{\Hex}{{\sc Hex}}
\newcommand{\dist}{\textup{dist}}
\newcommand{\cP}{\mathcal{P}}

\date{}

\renewenvironment{abstract}
{\small\vspace{-1em}
\begin{center}
\bfseries\abstractname\vspace{-.5em}\vspace{0pt}
\end{center}
\list{}{
\setlength{\leftmargin}{0.6in}%
\setlength{\rightmargin}{\leftmargin}}%
\item\relax}
{\endlist}

\tikzstyle{vertex_red}=[circle, inner sep=0, minimum size =4 pt, line width = 1pt, draw=red, fill=red]
\tikzstyle{vertex_blue}=[circle, inner sep=0, minimum size =4 pt, line width = 1pt, draw=blue, fill=blue]
\tikzstyle{vertex_black}=[circle, inner sep=0, minimum size =4 pt, line width = 1pt, draw=black, fill=black,text= white]
\tikzstyle{terminal_red}=[circle, inner sep=0, minimum size =6 pt, line width = 1pt, draw=red, fill=red]
\tikzstyle{terminal_blue}=[circle, inner sep=0, minimum size =6 pt, line width = 1pt, draw=blue, fill=blue]
\tikzstyle{vertex}=[circle,inner sep=0, minimum size =2.5 pt, line width = 1pt, draw=black, fill=black, text= white]

\title{The disjoint separators problem in graphs}

\author[1]{Thomas Delépine}
\author[2]{Florian Galliot}
\author[3]{Yannick Mogge}
\author[4]{Leandro Montero}
\author[5]{Nicolas Schivre}

\affil[1]{LIRMM, Université de Montpellier, CNRS, Montpellier, France}
\affil[2]{Aix-Marseille Université, CNRS, I2M, UMR 7373, Marseille, France}
\affil[3]{LIRIS, Université Claude Bernard Lyon 1, UMR5205, Villeurbanne, France}
\affil[4]{IMT Atlantique, LS2N - CNRS, La Chantrerie, Nantes, France}
\affil[5]{LIMOS, Université Clermont Auvergne, CNRS, Clermont-Ferrand, France}

\begin{document}

\renewcommand{\sectionautorefname}{Section}
\renewcommand{\subsectionautorefname}{Section}
\renewcommand{\subsubsectionautorefname}{Section}

\emergencystretch 2em

\maketitle

\begin{abstract}
We study the disjoint separators problem in graphs, an analogue of the famous disjoint paths problem. Given a graph $G$ and four pairwise disjoint subsets of vertices $S_r$, $T_r$, $S_b$, $T_b$, we ask whether there exist an $(S_r,T_r)$-separator and an $(S_b,T_b)$-separator which are disjoint. This is equivalent to coloring the vertices in red or blue, with $S_r \cup T_r$ in red and $S_b \cup T_b$ in blue, such that there is no red $(S_r,T_r)$-path and no blue $(S_b,T_b)$-path. On the one hand, we show that the disjoint separators problem is \NP-complete. We actually exhibit several \NP-complete restrictions of this problem, including planar graphs of bounded maximum degree, and graphs of bounded maximum degree when $|S_r|=|T_r|=|S_b|=|T_b|=1$. On the other hand, these hardness results turn out to be quite tight, as we provide a structural characterization and a polynomial-time algorithm for planar graphs when $|S_r|=|T_r|=|S_b|=|T_b|=1$. This has an interesting consequence about the popular board game \Hex: for the generalized game that may be played on any board, our result characterizes the planar boards on which draws are impossible, thus extending the well-known result about impossibility of draws on the standard commercialized board.
\end{abstract}

\section{Introduction}

\subsection{\Hex~and its generalization}

\Hex~is a classic two-player board game invented in 1942 \cite{politiken}. The board, shown in \autoref{fig:hexboard}(a), is a rhombus made of $11 \times 11$ hexagonal cells, with red borders on two opposite sides and blue borders on the other two sides. In turns, the players place a stone inside an unoccupied cell of their choice. The first player uses red stones, and she wins by connecting the two red borders with red stones, while the second player uses blue stones, and he wins by connecting the two blue borders with blue stones (where two cells are considered adjacent if they share a side). Equivalently, \Hex~can be seen as played on the graph shown in \autoref{fig:hexboard}(b). The vertices $s_r$ and $t_r$ are precolored in red, the vertices $s_b$ and $t_b$ are precolored in blue, and the players take turns coloring the other vertices. The first player colors vertices in red, and wins by getting a red $s_rt_r$-path, while the second player colors vertices in blue, and wins by getting a blue $s_bt_b$-path. The game \Hex~can be generalized to any graph with four designated terminals, two of them precolored in red and the other two in blue: the first player wants to build a path between the two red terminals, while the second player wants to build a path between the two blue terminals (if all vertices are colored without either thing happening, then the game is declared a draw). This game is part of a family of games which generalize {\em strong positional games} \cite{hefetz2014positional}, also called {\em Maker-Maker games}, to instances where the winning combinations are not necessarily the same for both players. This family has recently been introduced in \cite{EuroComb}: for example, if there is symmetry between the two colors (as is the case for \Hex), then a strategy-stealing argument ensures that the second player cannot have a winning strategy.

\begin{figure}[h]
    \centering
    \includegraphics[width=1\linewidth]{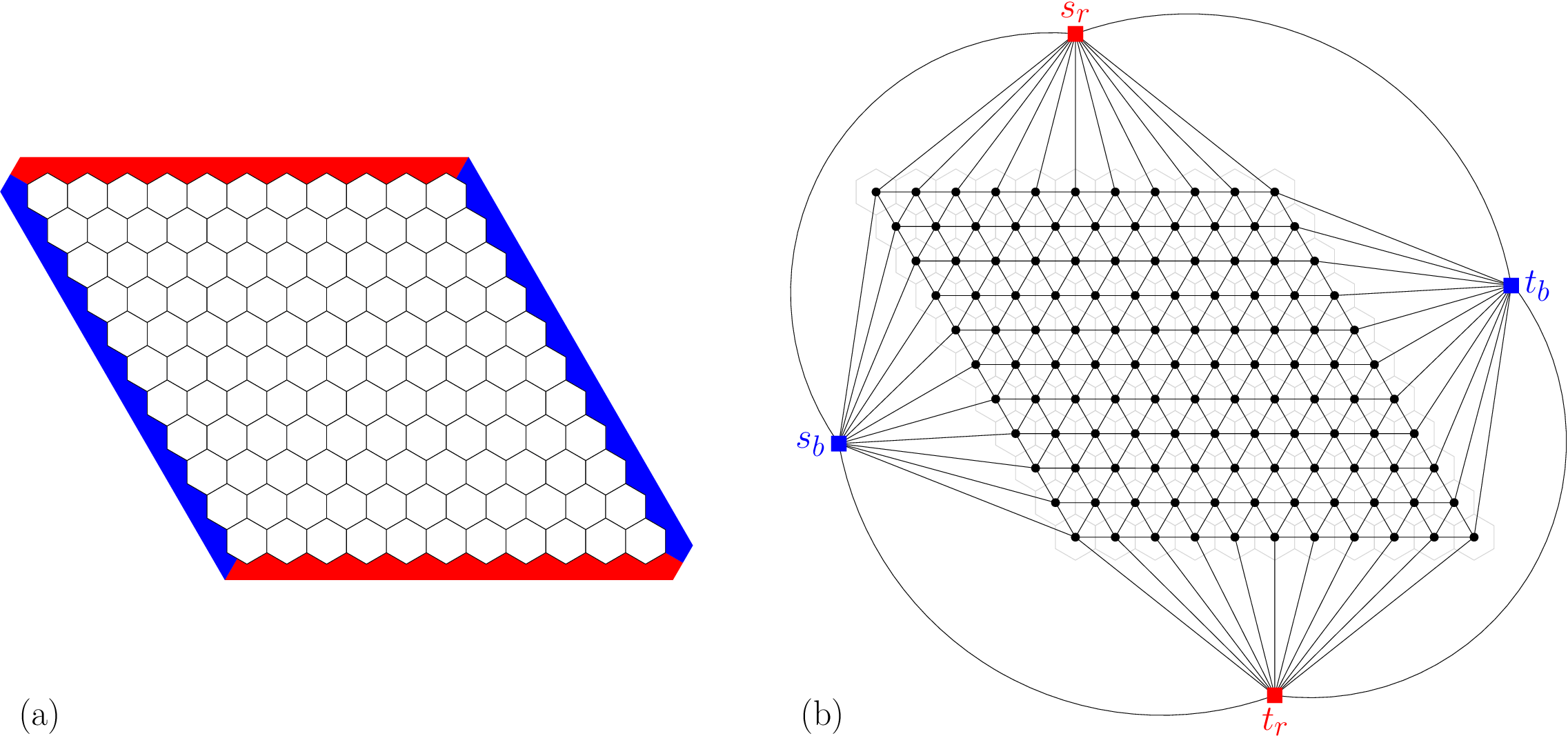}
    \caption{(a) The \Hex~board. (b) An alternative formulation of \Hex~as a coloring game connecting terminals in a graph. The four outside edges have only been added to fit the statement of the upcoming \autoref{thm:planar}.}
    \label{fig:hexboard}
\end{figure}

The original \Hex~game is sometimes presented with {\em Maker-Breaker} rules: the first player wins if she connects the two red borders with her red stones, while the second player wins if he prevents her from doing so (i.e., if she has not succedeed by the time all cells are occupied). Indeed, it turns out this apparent change of rules does not alter the game at all. This is the reason why, in the literature, \Hex~was first generalized as played on a graph with two red terminals (but no blue terminals), where the first player wins if she connects these two terminals and the second player wins if he prevents her from reaching that goal: this is the {\em Shannon switching game} played on vertices \cite{Shannon}. The fact that \Hex~is a Maker-Breaker game is an immediate consequence of two well-known facts about \Hex~\cite{Nash}:
\begin{enumerate}[label={(\arabic*)}]
	\item ``Winning is blocking'': given any red-blue coloring of the board, there cannot be both a red path between the two red terminals and a blue path between the two blue terminals.
	\item ``Blocking is winning'', i.e., draws are impossible: given any red-blue coloring of the board, there is either a red path between the two red terminals or a blue path between the two blue terminals.
\end{enumerate}
One can then wonder which structural features of the \Hex~board cause properties (1) and (2) above to hold. For instance, property (1) is implied by planarity and would hold just the same for any other planar board with opposite red/blue borders, whereas property (2) would fail if the cells were square rather than hexagonal for example (consider a checkerboard coloring of the square grid). For our generalized version of \Hex, played on arbitrary graphs with two red terminals and two blue terminals, can we get a structural characterization of properties (1) and (2) and/or algorithmic results on the problem of deciding whether these properties hold?

Both properties actually translate as natural problems on graphs. Property (1) corresponds to the well-known \textsc{DisjointPaths} problem: given a graph and two pairs of terminals $(s_r,t_r)$ and $(s_b,t_b)$, we want to know whether there exist an $s_rt_r$-path and an $s_bt_b$-path which are vertex-disjoint. This problem is known to be tractable \cite{Seymour1,Shiloach}, even when generalized to any fixed number of pairs of terminals \cite{Seymour2}. Property (2), on the other hand, corresponds to an analogue of the \textsc{DisjointPaths} problem about separators: given a graph and two pairs of terminals $(s_r,t_r)$ and $(s_b,t_b)$, we want to know whether there exist an $(s_r,t_r)$-separator and an $(s_b,t_b)$-separator (two sets of vertices whose removal disconnect the two corresponding terminals) which are disjoint. Indeed, seeing $s_r,t_r$ as precolored in red and $s_b,t_b$ as precolored in blue, two such separators exist if and only if there exists a red-blue coloring of the other vertices such that there is no red $s_rt_r$-path and no blue $s_bt_b$-path: the blue vertices then form an $(s_r,t_r)$-separator and the red vertices form an $(s_b,t_b)$-separator. This paper addresses this problem of finding disjoint separators which, to our knowledge, has not been studied before.

\subsection{The disjoint separators problem}

Given a (simple, finite, undirected) graph $G$ and disjoint subsets of vertices $S$ and $T$, an {\em $(S,T)$-separator} in $G$ is a set of vertices $X$ such that $X \cap (S \cup T) = \varnothing$ and whose removal disconnects $S$ from $T$, i.e., no connected component of $G-X$ intersects both $S$ and $T$. Obviously, an $(S,T)$-separator exists if and only if there is no edge between $S$ and $T$.

Let $G$ be a graph and let $S_r,T_r,S_b,T_b$ be four pairwise disjoint subsets of vertices. Vertices in $S_r \cup S_b$ (resp. in $T_r \cup T_b$) are called \emph{sources} (resp. \emph{targets}). More generally, vertices in $S_r \cup S_b \cup T_r \cup T_b$ are called {\em terminals}. We say that $(X,Y)$ is an {\em $(S_r,T_r,S_b,T_b)$-separator} if: $X$ is an $(S_r,T_r)$-separator, $Y$ is an $(S_b,T_b)$-separator, and $X \cap Y = \varnothing$. Note that, if such $(X,Y)$ exists, then we can choose it to partition the vertex set of $G$: indeed, we may simply add the vertices in $S_r \cup T_r$ to $Y$, those in $S_b \cup T_b$ to $X$, and the remaining vertices to any of the two arbitrarily.

This remark allows us to see our problem as a coloring problem, a viewpoint that we will use throughout the paper. The terminal vertices are precolored: those in $S_r \cup T_r$ in red, and those in $S_b \cup T_b$ in blue. We want to know if it is possible to color the non-terminal vertices using the colors red and blue, so that there is no red path between a vertex in $S_r$ and a vertex in $T_r$, and no blue path between a vertex in $S_b$ and a vertex in $T_b$. The set of red (resp. blue) vertices would then correspond to the separator $Y$ (resp. $X$).

\begin{definition}\label{def:cuts}
    The \generalcuts~decision problem is defined as follows. An instance is a tuple $(G,S_r,T_r,S_b,T_b)$ where $G$ is a graph and $S_r,T_r,S_b,T_b$ are pairwise disjoint subsets of vertices of $G$ such that there is no edge between $S_r$ and $T_r$ and no edge between $S_b$ and $T_b$. 
    The output is YES if and only if there exists an $(S_r,T_r,S_b,T_b)$-separator in $G$. Equivalently, the output is YES if and only if there exists a red-blue coloring of the vertices of $G$ such that: all vertices of $S_r \cup T_r$ are colored in red, all vertices of $S_b \cup T_b$ are colored in blue, and there is no monochromatic path between a source and a target of the same color.
\end{definition}

For convenience, we have specified in the previous definition that there is no edge between $S_r$ and $T_r$ and no edge between $S_b$ and $T_b$, as otherwise we would have a trivial NO-instance.

The subproblem where there is exactly one red source, one red target, one blue source and one blue target (as in the aforementioned generalization of \Hex) is of particular interest, hence the following definition.

\begin{definition}
    The \cuts~decision problem is the restriction of \generalcuts~to instances $(G,S_r,T_r,S_b,T_b)$ such that $|S_r|=|T_r|=|S_b|=|T_b|=1$.
\end{definition}

We will write instances of \cuts~as $(G,s_r,t_r,s_b,t_b)$, dropping the curly brackets for singletons to alleviate notations. Throughout the paper, $n$ denotes the number of vertices of the graph under consideration.

\subsection{Overview of the results}

In this work, we focus in particular on the planar case of \cuts. On the positive side, we show that \cuts~can be solved in polynomial time on planar graphs. For this, we first provide a structural characterization of YES-instances when the four terminals form a cycle, which yields a quadratic-time algorithm. This answers the question of impossibility of draws for \Hex~generalized to any planar board (if we do not require balance between both players' number of ``moves''). We then extend the results to arbitrary planar instances of \cuts~by reducing to the case where the four terminals form a cycle.

On the negative side, we prove that \generalcuts~is \NP-complete even when restricted to planar graphs of maximum degree at most $5$, and it remains so when adding the restriction that the number of terminals is sublinear in the order of the graph.  
We then derive several corollaries about the hardness of \cuts. We show that it is \NP-complete even when restricted to graphs of maximum degree at most 7, or to graphs where the removal of the four terminals yields a planar graph of maximum degree at most 3.

The paper is organized as follows. In \autoref{sec:prelim}, we prove several preliminary properties of \generalcuts, including simple reductions and structural lemmas. \autoref{sec:npc} is devoted to hardness results: we prove that \generalcuts~is \NP-complete, and derive several corollaries. 
In \autoref{sec:planar}, we study the planar case of \cuts~in detail and present a polynomial-time algorithm, based on a structural characterization of YES-instances when the four terminals form a cycle. 
Finally, in \autoref{sec:conclusions}, we discuss concluding remarks and directions for future research.

\section{Preliminaries}\label{sec:prelim}

\subsection{Some straightforward instances}

We start by noting that YES-instances of \generalcuts~are stable under taking subgraphs. We state the contrapositive for future reference.

\begin{proposition}\label{prop:subgraph}
    Let $(G,S_r,T_r,S_b,T_b)$ be an instance of \generalcuts, and let $G_0$ be a subgraph of $G$ whose vertex set includes $S_r \cup T_r \cup S_b \cup T_b$. If $(G_0,S_r,T_r,S_b,T_b)$ is a NO-instance of \generalcuts~(i.e., any red-blue coloring of the non-terminal vertices in $G_0$ yields a monochromatic source-to-target path), then $(G,S_r,T_r,S_b,T_b)$ also is a NO-instance.
\end{proposition}

\begin{proof}
    Consider a red-blue coloring of the non-terminal vertices in $G$. Since $(G_0,S_r,T_r,S_b,T_b)$ is a NO-instance, the subgraph induced by $G_0$ contains a red path between $S_r$ and $T_r$ or a blue path between $S_b$ and $T_b$. Hence, $G$ also contains such a path.
\end{proof}

In the following proposition, $\dist_G(A,B)$ refers to the length of a shortest path in $G$ between a vertex $a \in A$ and a vertex $b \in B$ (or $\infty$ if no such path exists, even though this case is not relevant to us).

\begin{proposition}\label{prop:distance}
	Let $(G,S_r,T_r,S_b,T_b)$ be an instance of \generalcuts. Suppose that $\dist_G(S_r,S_b) \geq 3$, or $\dist_G(S_r,T_b) \geq 3$, or $\dist_G(T_r,S_b) \geq 3$, or $\dist_G(T_r,T_b) \geq 3$. Then, $(G,S_r,T_r,S_b,T_b)$ is a YES-instance.
\end{proposition}

\begin{proof}
	Assume that $\dist_G(S_r,S_b) \geq 3$ (the other cases are analogous). Then, we color the non-terminal vertices that are neighbors of a red source in blue, and the other non-terminal vertices in red. In this way, all neighbors of red sources are colored in blue, and, since $\dist_G(S_r,S_b) \geq 3$, all neighbors of blue sources are colored in red. This clearly ensures that there is no monochromatic path between a source and a target of the same color.
\end{proof}

\subsection{Useful reductions}

\begin{proposition}
    \generalcuts~admits a linear-time reduction to \cuts.
\end{proposition}

\begin{proof}
    Let $(G,S_r,T_r,S_b,T_b)$ be an instance of \generalcuts. Let $G_0$ be the graph obtained from $G$ by adding four new vertices $s_r,t_r,s_b,t_b$ and adding: all edges between $\{s_r\}$ and $S_r \cup S_b \cup T_b$, all edges between $\{t_r\}$ and $T_r \cup S_b \cup T_b$, all edges between $\{s_b\}$ and $S_b \cup S_r \cup T_r$, and all edges between $\{t_b\}$ and $T_b \cup S_r \cup T_r$. 
    We claim that the instances $(G_0,s_r,t_r,s_b,t_b)$ and $(G,S_r,T_r,S_b,T_b)$ are equivalent. Indeed, in the former instance, every vertex $u \in S_r \cup T_r$ must be colored in red because of the path $s_but_b$, and every vertex $v \in S_b \cup T_b$ must be colored in blue because of the path $s_rvt_r$. From there, the existence of a red path between $s_r$ and $t_r$ (resp. between $s_b$ and $t_b$) is equivalent to the existence of a red path between $S_r$ and $T_r$ (resp. between $S_b$ and $T_b$).
\end{proof}

Note that the previous reduction does not necessarily preserve planarity, unlike the next result about edge contractions.

\begin{proposition}\label{prop:contraction}
	Let $(G,S_r,T_r,S_b,T_b)$ be an instance of \generalcuts. If two vertices $u,v \in S_r$ are adjacent, then, denoting by $G'$ the graph obtained through contracting the edge $uv$ and by $w$ the new contracted vertex, the instances $(G,S_r,T_r,S_b,T_b)$ and $(G',(S_r \setminus \{u,v\}) \cup \{w\},T_r,S_b,T_b)$ are equivalent. Analogous statements hold for $T_r$, $S_b$ and $T_b$. 
\end{proposition}

\begin{proof}
	Note that the non-terminal vertices are the same for both instances. Given a red-blue coloring of the non-terminal vertices, we clearly have a red path between $w$ and some vertex in $T_r$ in the second instance if and only if we have a red path between $u$ or $v$ and some vertex in $T_r$ in the first instance.
\end{proof}

In particular, using \autoref{prop:contraction}, we can always reduce to the case where $S_r$, $T_r$, $S_b$ and $T_b$ are stable sets.

Finally, we show that \cuts~admits a reduction to the 2-connected case. Recall that a graph $G$ is called {\em $k$-vertex-connected} (or {\em $k$-connected} for short) if at least $k$ vertices must be removed from $G$ to disconnect it, and that a {\em cut vertex} is a vertex whose removal disconnects the graph.

\begin{proposition}\label{prop:2-connected}
    There is a quadratic-time algorithm which, given an instance $(G,s_r,t_r,s_b,t_b)$ of \cuts, either solves it or outputs an equivalent instance $(G_0,s'_r,t'_r,s'_b,t'_b)$ satisfying all the following properties:
    \begin{itemize}[nolistsep,noitemsep]
        \item $G_0$ is 2-connected;
        \item $G_0$ is an induced subgraph of $G$;
        \item If $s_rs_bt_rt_b$ is a cycle in $G$, then $(s'_r,t'_r,s'_b,t'_b)=(s_r,t_r,s_b,t_b)$.
    \end{itemize}
\end{proposition}

\begin{proof}

    Consider an instance $(G,s_r,t_r,s_b,t_b)$. Recall that, by \autoref{def:cuts}, $s_r$ and $t_r$ are not adjacent, and neither are $s_b$ and $t_b$.
    
    Let us note that, if $s_r$, $t_r$, $s_b$ and $t_b$ do not all belong to the same connected component of $G$, then we have a YES-instance: indeed, for each connected component $C$, simply give all the non-terminal vertices of $C$ the same color, blue if $\{s_r,t_r\} \subseteq C$ or red otherwise.
    
    Therefore, we may assume that $G$ is connected, by only considering the connected component containing $s_r$, $t_r$, $s_b$ and $t_b$. We now explain how to eliminate each cut vertex in linear time, so that we get a 2-connected graph in quadratic time. Let $v$ be a cut vertex of $G$, and note that finding $v$ and computing the connected components of $G-v$ can easily be done in linear time.
    
    First, suppose that $v \in \{s_r,t_r,s_b,t_b\}$. Without loss of generality, assume $v=s_r$. Let $C$ be the connected component of $G-s_r$ that contains $t_r$. If $\{s_b,t_b\} \subseteq C$, then $(G[C \cup \{s_r\}],s_r,t_r,s_b,t_b)$ is an equivalent instance. 
    If $\{s_b,t_b\} \not\subseteq C$, then we have a YES-instance: simply color the non-terminal vertices of $C$ in blue and the other non-terminal vertices in red.

    We now suppose that $v \not\in \{s_r,t_r,s_b,t_b\}$. For $x \in \{s_r,t_r,s_b,t_b\}$, let $C(x)$ be the connected component of $G-v$ that contains $x$. Using the symmetries, we only need to consider the following four cases, with \autoref{fig:2-connectedness-reduction} providing some illustrations.
    
    \begin{enumerate}[label={\arabic*)}]
        \item First case: $C(s_r)=C(t_r)=C(s_b)=C(t_b)$.

        Clearly, $(G[C(s_r) \cup \{v\}],s_r,t_r,s_b,t_b)$ is an equivalent instance. 
        Note that we are always in this case if $s_rs_bt_rt_b$ is a cycle, so we do keep the same four terminals in that case as required by the statement of this proposition.

        \item Second case: $C(s_r) \neq C(t_r)$, and $C(s_b) \neq C(t_b)$.

        If $s_r$, $t_r$, $s_b$ and $t_b$ are all neighbors of $v$ in $G$, then we obviously have a NO-instance since there will be either a red path $s_rvt_r$ or a blue path $s_bvt_b$. Otherwise, we have a YES-instance: indeed, if $s_r$ or $t_r$ (resp. $s_b$ or $t_b$) is not a neighbor of $v$, simply color $v$ in red (resp. blue) and the other non-terminal vertices in blue (resp. red).

        \item Third case: $C(s_r)=C(t_r)$, $C(s_b) \neq C(s_r)$, and $C(t_b) \neq C(s_r)$.

        Again, we assume that $s_r$, $t_r$, $s_b$ and $t_b$ are not all neighbors of $v$ in $G$. We then have a YES-instance. Indeed, if $s_r$ or $t_r$ is not a neighbor of $v$, then color all the non-terminal vertices of $C(s_r)$ in blue and the other non-terminal vertices (this includes $v$) in red. Otherwise, if $s_b$ (resp. $t_b$) is not a neighbor of $v$, then color all the non-terminal vertices of $C(s_b)$ (resp. $C(t_b)$) in red and the other non-terminal vertices (this includes $v$) in blue.

        \item Fourth case: $C(s_r)=C(t_r)=C(t_b)$, and $C(s_b) \neq C(t_b)$.

        Again, we assume that $s_r$, $t_r$, $s_b$ and $t_b$ are not all neighbors of $v$ in $G$.
        
        \begin{enumerate}[label={\alph*)},nolistsep]
            \item If either $s_r$ or $t_r$ is not a neighbor of $v$, then we have a YES-instance: simply color $v$ in red and the other non-terminal vertices in blue.
            \item If $s_b$ is not a neighbor of $v$, then we also have a YES-instance: 
            color all non-terminal vertices of $C(s_b)$ in red and the other non-terminal vertices (this includes $v$) in blue.
            \item Finally, suppose that $s_r$, $t_r$ and $s_b$ are neighbors of $v$, but $t_b$ is not. Since $s_r$ and $t_r$ are neighbors of $v$, coloring $v$ in blue is forced. Since $s_b$ is a neighbor of $v$, for any coloring in which $v$ is blue, there is a blue $s_bt_b$-path in $G$ if and only if there is a blue $vt_b$-path in $G[C(t_b) \cup \{v\}]$. Therefore, $(G[C(t_b) \cup \{v\}],s_r,t_r,v,t_b)$ is an equivalent instance. \qedhere
        \end{enumerate}
    \end{enumerate}
\end{proof}

\begin{figure}
    \centering
    \includegraphics[width=1\linewidth]{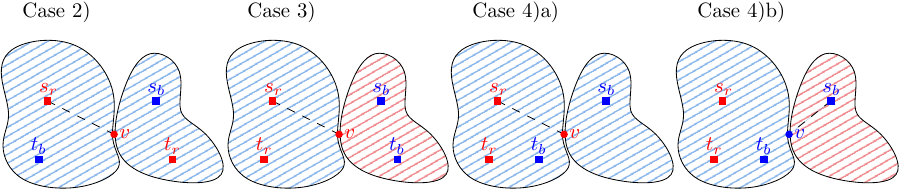}
    \caption{Illustration of Cases 2), 3), 4)a) and 4)b) from the proof of \autoref{prop:2-connected}. In Cases 2), 3) and 4)a), we assume without loss of generality that $s_{r} v$ is not an edge. In Case 4)b), we assume without loss of generality that $s_{b} v$ is not an edge.
    }
    \label{fig:2-connectedness-reduction}
\end{figure}

\section{\NP-completeness results}\label{sec:npc}

In this section, we will prove hardness results regarding \generalcuts~in planar graphs of bounded maximum degree and other variants. If $(G, S_r, T_r, S_b, T_b)$ is a YES-instance of \generalcuts, with $(X,Y)$ as a solution, then the description of $(X,Y)$ constitutes a certificate verifiable in polynomial time, hence the following proposition.

\begin{proposition}\label{prop:np}
    \generalcuts~is in \NP.
\end{proposition}

\begin{proof}
A certificate is a red-blue coloring of the vertices. Verifying that no red path connects $S_r$ to $T_r$ and no blue path connects $S_b$ to $T_b$ can be done using standard graph search in polynomial time, so the problem is clearly in \NP.
\end{proof}

The certificate verification procedure runs in polynomial time independently of the sizes of $S_r$, $T_r$, $S_b$, and $T_b$, as well as of any structural restrictions imposed on $G$. Hence, \cuts~is also in \NP, and the same holds for all variants of \generalcuts~and \cuts~considered in this paper. The following proofs of \NP-completeness will therefore only consist in proofs of \NP-hardness.  

The proof of our first result is a reduction from \planarsat.
Let $\varphi$ be a CNF formula with all clauses of size $3$. Let $G_\varphi$ be the graph with vertex set 
\[\{v_x \mid x\textrm{ is a variable of } \varphi\} \cup \{v_C \mid C\textrm{ is a clause of } \varphi\}\]
and edge set 
\[\{v_xv_C \mid x \textrm{ is a variable appearing in the clause }C \textrm{ of }\varphi\}.\] 
The graph $G_\varphi$ is called the \emph{variable-clause incidence graph} of $\varphi$. It is well known that \planarsat~(the restriction of \sat~to instances whose variable-clause incidence graph is planar) is \NP-complete \cite{Lichtenstein1982}. In what follows, we reduce \planarsat~to \generalcuts, proving that our problem is \NP-hard even for planar instances of bounded maximum degree.

\begin{theorem}\label{thm:disjoint_cuts_NPC}
     \generalcuts~is \NP-complete even when restricted to planar graphs of maximum degree at most $5$. 
 \end{theorem}

 \begin{proof} 
    We perform a reduction from \planarsat. Let $\varphi$ be an instance of \planarsat, without any clause containing both a variable and its negation since that would be a trivial clause. We assume each clause $C = (\ell_1, \ell_2, \ell_3)$ in $\varphi$ to be ordered with the positive literals to the left of the negative literals, i.e., either $C$ is monotone or $\ell_1$ is positive and $\ell_3$ is negative. We fix a planar embedding of the variable-clause incidence graph $G_\varphi$. From $\varphi$ and $G_\varphi$, we now construct an instance $(G, S_r, T_r, S_b, T_b)$ of \generalcuts~in three steps. The reader can find an illustration of the whole reduction in \autoref{fig:example_reduction}. 
     
     \begin{enumerate}[label={(\arabic*)}]

     \item For every variable $x$ appearing in $\varphi$, let $k$ be the number of clauses of $\varphi$ containing $x$. We add a copy of $\textrm{VAR}_x^k$ from \autoref{fig:var_gadget} to $G$.

     \begin{figure}[h]
        \centering
        \includegraphics[scale=1]{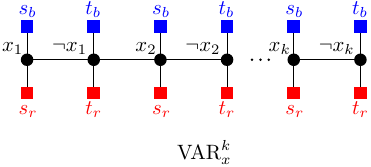}
        \caption{The variable gadget used in the proof of \autoref{thm:disjoint_cuts_NPC}. Vertices in $S_r$ (resp. $T_r$, $S_b$, $T_b$) are generically labeled $s_r$ (resp. $t_r$, $s_b$, $t_b$) in all figures.
        }
        \label{fig:var_gadget}
    \end{figure}
    
    \item Then, for every clause $C = (\ell_1, \ell_2, \ell_3)$ in $\varphi$, we add a copy of $\textrm{CL}_{(\ell_1, \ell_2, \ell_3)}$ from \autoref{fig:cla_gadget} to $G$. 

    \begin{figure}[h]
         \centering
         \includegraphics[scale=1]{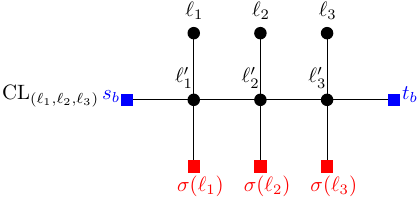}
         \caption{The clause gadget used in the proof of \autoref{thm:disjoint_cuts_NPC}. The label $\sigma(\ell_i)$ is defined as $t_r$ if $\ell_i$ is a positive literal or $s_r$ otherwise. Recall that the clause $(\ell_1, \ell_2, \ell_3)$ is ordered with the positive literals to the left of the negative literals.}
         \label{fig:cla_gadget}
    \end{figure}
    
    \item Finally, for every variable $x$ of $\varphi$, let $v_{C_1},\dots,v_{C_k}$ be the neighbors of $v_x$ in $G_\varphi$ respecting the circular ordering given by the planar embedding of $G_\varphi$ around $v_x$ going clockwise. For every $i \in \{1,\dots,k\}$, let $\ell_j$ be the unique literal in $C_i$ such that $\ell_j \in \{x, \neg x\}$. If $\ell_j = x$, then we identify the vertices $\ell_j$ from $\textrm{CL}_{(\ell_1, \ell_2, \ell_3)}$ and $x_i$ from $\textrm{VAR}_x^k$. If $\ell_j = \neg x$, then we identify the vertices $\ell_j$ from $\textrm{CL}_{(\ell_1, \ell_2, \ell_3)}$ and $\neg x_i$ from $\textrm{VAR}_x^k$.

    \end{enumerate}

    \begin{figure}[H]
    \centering
    \includegraphics[width=1\linewidth]{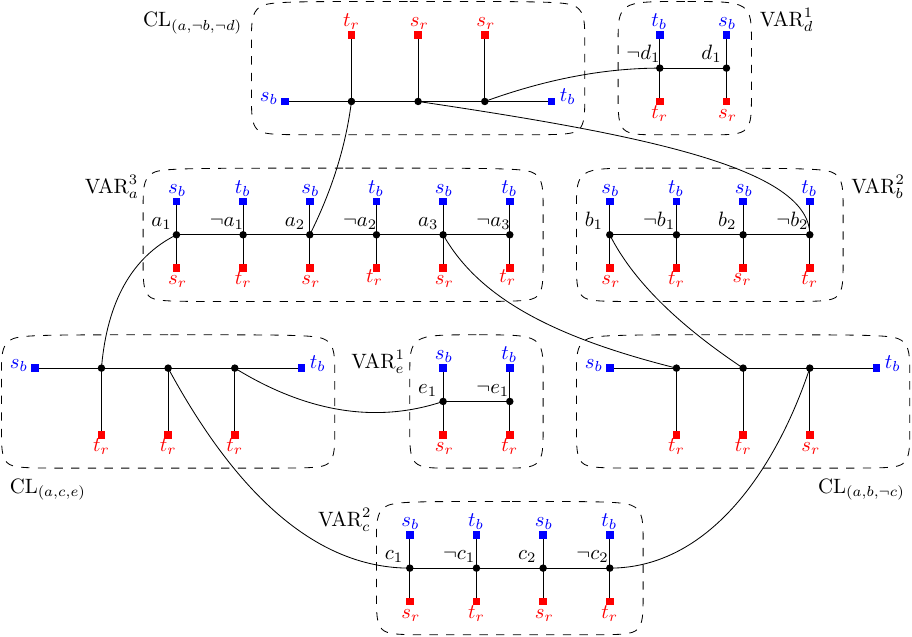}
    \caption{The instance of \generalcuts~constructed by the reduction used in the proof of \autoref{thm:disjoint_cuts_NPC} from the instance $\varphi = (a\vee b\vee \neg c) \wedge (a \vee \neg b \vee \neg d) \wedge (a \vee c \vee e)$ of \planarsat.
    }
    \label{fig:example_reduction}
\end{figure}

    Clearly, $G$ has maximum degree at most 5. It remains to prove that $G$ is planar. With Step $(3)$ from the construction, and since each variable gadget is a tree, no crossing is created when replacing each vertex $v_x$ by its variable gadget $\textrm{VAR}_x^k$. As for the clause gadgets, we provide an explicit planar embedding. Let $C = (\ell_1, \ell_2, \ell_3)$ be a clause of $\varphi$. Since $v_C$ is of degree $3$, there are only two possible circular orderings for $\{\ell_1, \ell_2, \ell_3\}$: the clockwise ordering, or the anti-clockwise ordering. For both cases, the planar embedding is depicted in \autoref{fig:planar_embeddingV2}.
    
    \begin{figure}[ht]
        \centering
        \includegraphics[width=0.8\linewidth]{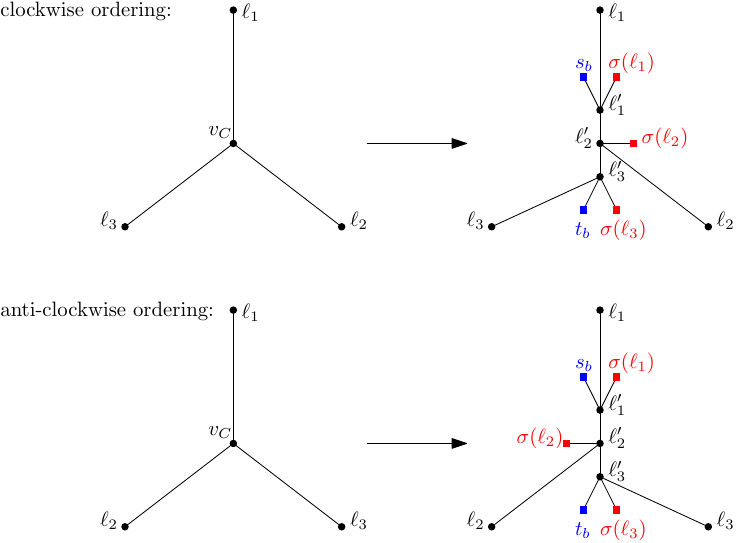}
        \caption{Planar embedding of $\textrm{CL}_{(\ell_1, \ell_2, \ell_3)}$ depending on the cyclic ordering around $v_C$ in $G$, where $C = (\ell_1, \ell_2, \ell_3)$ is an ordered clause.}
        \label{fig:planar_embeddingV2}
    \end{figure}
        
    We now claim that $\varphi$ is satisfiable if and only if there exists an $(S_r,T_r, S_b, T_b)$-separator in $G$, i.e., if and only if there exists a red-blue coloring of the non-terminal vertices in $G$ such that there is no monochromatic path between a source and a target of the same color.
     \begin{itemize}
          \item First, suppose that $\varphi$ is satisfiable. We fix an assignment of truth values to the variables which satisfies $\varphi$. Then, for each variable $x$ of $\varphi$, we color $x_1, \dots, x_k$ from the variable gadget $\textrm{VAR}_x^k$ in blue if $x$ is true or in red if $x$ is false, and we color $\neg x_1, \dots, \neg x_k$ with the opposite color. One can verify that this coloring creates no monochromatic path between a source and a target in the variable gadgets. Also note that, as a consequence, for each clause $C = (\ell_1, \ell_2, \ell_3)$ and each $i \in \{1,2,3\}$, the vertex $\ell_i$ from the clause gadget $\textrm{CL}_{(\ell_1, \ell_2, \ell_3)}$ is now colored in blue if and only if the literal $\ell_i$ is true (in particular, at least one of $\ell_1$, $\ell_2$ or $\ell_3$ is blue). On the other hand, the vertices $\ell'_1$, $\ell'_2$ and $\ell'_3$ are still uncolored. To color them, we follow the rules provided by \autoref{tab:coloring_rules}. It can easily be checked that this creates no monochromatic path between a source and a target in the clause gadgets. 
     \begin{table}[H]
             \centering
             \begin{tabular}{|c|c|c|c|c|c|c|l}
                    \cline{1-7} $\ell_1$ & $\ell_2$ & $\ell_3$ &  & $\ell_1'$ & $\ell_2'$ & $\ell_3'$ & \\
                 \cline{1-7} \textcolor{blue}{blue} & \textcolor{blue}{blue} & \textcolor{blue}{blue} & $\rightarrow$ & \textcolor{red}{red} & \textcolor{blue}{blue} & \textcolor{red}{red} & \\
                 \cline{1-7} \textcolor{red}{red} & \textcolor{blue}{blue} & \textcolor{blue}{blue} & $\rightarrow$ & \textcolor{blue}{blue} & \textcolor{red}{red} & \textcolor{red}{red} & if $C$ is monotone \\
                 \textcolor{red}{red} & \textcolor{blue}{blue} & \textcolor{blue}{blue} & $\rightarrow$ & \textcolor{blue}{blue} & \textcolor{red}{red} & \textcolor{blue}{blue} & if $C$ is not monotone \\
                 \cline{1-7} \textcolor{blue}{blue} & \textcolor{red}{red} & \textcolor{blue}{blue} & $\rightarrow$ & \textcolor{red}{red} & \textcolor{blue}{blue} &  \textcolor{red}{red} &  \\
                 \cline{1-7} \textcolor{blue}{blue} & \textcolor{blue}{blue} & \textcolor{red}{red} & $\rightarrow$ & \textcolor{red}{red} & \textcolor{red}{red} & \textcolor{blue}{blue} & if $C$ is monotone \\
                 \textcolor{blue}{blue} & \textcolor{blue}{blue} & \textcolor{red}{red} & $\rightarrow$ & \textcolor{blue}{blue} & \textcolor{red}{red} & \textcolor{blue}{blue} & if $C$ is not monotone \\
                 \cline{1-7} \textcolor{red}{red} & \textcolor{red}{red} & \textcolor{blue}{blue} & $\rightarrow$ & \textcolor{blue}{blue} & \textcolor{blue}{blue} & \textcolor{red}{red} &  \\
                 \cline{1-7} \textcolor{red}{red} & \textcolor{blue}{blue} & \textcolor{red}{red} & $\rightarrow$ & \textcolor{blue}{blue} & \textcolor{red}{red} & \textcolor{blue}{blue} &  \\
                 \cline{1-7} \textcolor{blue}{blue} & \textcolor{red}{red} & \textcolor{red}{red} & $\rightarrow$ & \textcolor{red}{red} & \textcolor{blue}{blue} & \textcolor{blue}{blue} &  \\
                 \cline{1-7}
             \end{tabular}
             \caption{The rules used to color $\ell_1'$, $\ell_2'$ and $\ell_3'$ for every (ordered) clause $C = (\ell_1, \ell_2, \ell_3)$.}
             \label{tab:coloring_rules}
         \end{table}
         \item Conversely, suppose that there exists a red-blue coloring of the non-terminal vertices that creates no monochromatic path between a source and a target. 
         Now, in any variable gadget $\textrm{VAR}_x^k$, the vertices $x_1, \dots, x_k$ must share the same color and the vertices $\neg x_1, \dots, \neg x_k$ must share the other color, otherwise a monochromatic source-to-target path would be created. This observation allows us to construct an assignment of truth values to the variables of $\varphi$ by setting $x$ to be true if and only if the vertices $x_1, \dots, x_k$ are colored in blue. 
         Suppose for a contradiction that $\varphi$ is not satisfied, i.e., there exists a clause $C = (\ell_1, \ell_2, \ell_3)$ containing only false literals. This means that the vertices $\ell_1$, $\ell_2$ and $\ell_3$ are colored in red. Recall that, by definition of $\sigma$ in the construction of $\textrm{CL}_{(\ell_1, \ell_2, \ell_3)}$, we have that for every $i \in \{1, 2, 3\}$: if $\ell_i$ is a neighbor of a vertex in $S_r$ then $\ell_i'$ is a neighbor of a vertex in $T_r$, and if $\ell_i$ is a neighbor of a vertex in $T_r$ then $\ell_i'$ is a neighbor of a vertex in $S_r$. Therefore, if there exists $i \in \{1, 2, 3\}$ such that $\ell'_i$ is colored in red, then a red path connects a red source to a red target across the clause gadget $\textrm{CL}_{(\ell_1, \ell_2, \ell_3)}$ and the variable gadget $\textrm{VAR}_x^k$ (where $x$ is the unique variable such that $\ell_i \in \{x,\neg x\}$), a contradiction. 
         Thus, the three vertices $\ell_1'$, $\ell_2'$ and $\ell_3'$ must be colored in blue. However, they then connect the blue source and the blue target in the clause gadget $\textrm{CL}_{(\ell_1, \ell_2, \ell_3)}$, which is a contradiction as well.
    \end{itemize}
    This equivalence proves that \generalcuts~is \NP-hard, even when restricted to planar graphs of maximum degree at most $5$. Since \generalcuts~is in \NP~by \autoref{prop:np}, we get the desired \NP-completeness result.
 \end{proof}

From \autoref{thm:disjoint_cuts_NPC} and its proof, we can deduce the following corollaries on the problem with only four terminals, namely, \cuts.

\begin{corollary}\label{coro:cuts_npc_in_distance2planar_4}
    \cuts~is \NP-complete even when restricted to instances $(G, s_r, t_r, s_b, t_b)$ where the graph $G -\{s_r, t_r, s_b, t_b\}$ is planar and of maximum degree at most $3$.
\end{corollary}

\begin{proof}
    We adapt the proof of \autoref{thm:disjoint_cuts_NPC}, using the same construction with the addition of a step $(4)$ consisting in identifying every vertex in $S_r$ (resp. $T_r$, $S_b$, $T_b$) into a single vertex $s_r$ (resp. $t_r$, $s_b$, $t_b$). The remainder of the proof follows the same arguments.
\end{proof}

\begin{corollary}\label{coro:npc}
    \cuts~is \NP-complete even when restricted to graphs of maximum degree at most $7$. 
\end{corollary}

\begin{proof}
    Again, we adapt the proof of \autoref{thm:disjoint_cuts_NPC}. In the reduction presented in the proof of \autoref{thm:disjoint_cuts_NPC}, many terminals are used. However, \cuts~requires exactly one terminal of each type. 
    To obtain this while maintaining boundedness of the maximum degree, 
    we will need some non-terminal vertices to simulate additional sources and targets. For this, we will add a ``copy gadget''. Consider the graph $G_1$ from \autoref{fig:cop_gadget}. The only way to color $a,b,c,d$ without creating a monochromatic source-to-target path is to color $a,b$ in blue and $c,d$ in red.
    Notice that, since $a$ is colored in blue and is a neighbor of a blue source, if there is a blue path from $a$ to a blue target, then there is a blue path from a blue source to a blue target. 
    Therefore, even though $a$ is not a terminal, $a$ behaves like a blue source. This observation works similarly for $b$ (resp. $c$,  $d$) which behaves like a blue target (resp. red source, red target). By chaining this construction (as illustrated with the graph $G_2$ from \autoref{fig:cop_gadget}), we can create as many vertices that behave like red or blue sources or targets as we need, all with maximum degree $6$. 

    The proof then proceeds as follows. We first assume that we can use as many terminals as we need and follow steps $(1)$, $(2)$ and $(3)$ from the proof of \autoref{thm:disjoint_cuts_NPC}. We then add a fourth step:
    \begin{itemize}
        \item[$(4)$] Add a sufficiently large copy of the graph $G_2$ from \autoref{fig:cop_gadget} ($6m$ layers are enough, where $m$ is the number of clauses of $\varphi$, as in the worst case we need $3$ layers per clause and $1$ layer per literal). Then, we consider each terminal $v$ that was created during the first three steps. Note that $v$ is currently of degree $1$. We turn $v$ into a non-terminal vertex by identifying it with a vertex from $G_2$ that is still of degree at most $6$ and that simulates a red source (resp. red target, blue source, blue target) if $v$ was assumed to belong to $S_r$ (resp. $T_r$, $S_b$, $T_b$). Hence, the maximum degree remains at most $7$.  
    \end{itemize}
    The remainder of the proof follows the same arguments.
\end{proof}

\begin{figure}[H]
    \centering
    \includegraphics[scale=1]{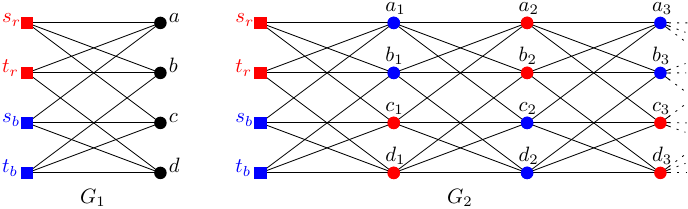}
    \caption{A construction allowing us to have as many vertices behaving like red or blue sources or targets as desired. The four vertices on the left of $G_2$ are the real terminals. The other vertices are not terminals, but we represent them as colored according to their unique available color, as explained in the proof of \autoref{coro:npc}.}
    \label{fig:cop_gadget}
\end{figure}

We end this section with a remark about the number of terminals in instances of \NP-hard restrictions of \generalcuts. In the construction from the proof of \autoref{thm:disjoint_cuts_NPC}, we have $|S_r| + |T_r| + |S_b| + |T_b| = \Theta(n)$. However, the proportion of terminal vertices can actually be made arbitrarily small, in the following sense.

\begin{corollary}\label{coro:epsilon_NPC}
    For every constant $\varepsilon > 0$, \generalcuts~is \NP-complete even when restricted to instances $(G,S_r,T_r,S_b,T_b)$ where $G$ is a planar graph of maximum degree at most $5$ and $|S_r| + |T_r| + |S_b| + |T_b| = O(n^\varepsilon)$. 
\end{corollary}

\begin{proof}
    We perform a reduction from \generalcuts~restricted to planar graphs of maximum degree at most $5$ (see \autoref{thm:disjoint_cuts_NPC}). Let $(G, S_r, T_r, S_b, T_b)$ be an instance of this problem, and let $n$ be the order of $G$. If $|S_r| + |T_r| + |S_b| + |T_b| \le n^\varepsilon$, then the reduction is just a copy of $(G, S_r, T_r, S_b, T_b)$. Therefore, assume that $|S_r| + |T_r| + |S_b| + |T_b| > n^\varepsilon$. We construct a new instance $(G', S_r, T_r, S_b, T_b)$, where $G'$ is obtained from $G$ by adding an independent set of order 
    \[k =  \left\lceil\left(|S_r| + |T_r| + |S_b| + |T_b|\right)^{1/\varepsilon}\right\rceil - n > 0.\] It is a straightforward observation that $(G, S_r, T_r, S_b, T_b)$ is a YES-instance if and only if $(G', S_r, T_r, S_b, T_b)$ is a YES-instance. Since the added independent set has size $k \ge (|S_r| + |T_r| + |S_b| + |T_b|)^{1/\varepsilon} - n$, we have $|S_r| + |T_r| + |S_b| + |T_b| \le (k + n)^\varepsilon = |V(G')|^\varepsilon$. Moreover, since $k$ is polynomial in $n$, the reduction is polynomial.
\end{proof}

In the previous argument, the size of the added independent set needs to be polynomial for the reduction to be polynomial. It is thus natural to wonder if we can improve the statement of \autoref{coro:epsilon_NPC}.

\begin{question}
Is \generalcuts~\NP-hard for instances $(G,S_r,T_r,S_b,T_b)$ such that $G$ is planar and $|S_r| + |T_r| + |S_b| + |T_b| = n^{o(1)}$?
\end{question}

\section{\cuts~in planar graphs}\label{sec:planar}

In this section, we will provide a structural characterization and a polynomial-time algorithm for \cuts~in planar graphs. 

Throughout this section, given a planar embedding of a 2-connected planar graph and a cycle $C$ of this graph, we say that a vertex/edge/cycle/face is {\em inside} (resp. {\em outside}) $C$ if all its vertices and edges are in the interior (resp. exterior) of $C$ or part of $C$ itself. To exclude vertices and edges of $C$, we use the term {\em strictly inside} $C$ or {\em strictly outside} $C$. In particular, the outer face is outside any given cycle. 
Finally, given three cycles $C$, $C_1$ and $C_2$, we say that $C$ is {\em between} $C_1$ and $C_2$ if either: $C_1$ is inside $C$ and $C_2$ is outside $C$, or if $C_2$ is inside $C$ and $C_1$ is outside $C$. Note that this definition does not mention the adverb ``strictly'', so that $C$ may be between $C_1$ and $C_2$ even if $C$ shares vertices or edges with $C_1$ and $C_2$ (we may even have $C=C_1=C_2$). 
We will use the same vocabulary for faces instead of cycles, in which case it is implicit that we consider their boundary cycle. See \autoref{fig:inside_outside} for an illustration.

\begin{figure}[h]
    \centering
    \includegraphics[scale=.64]{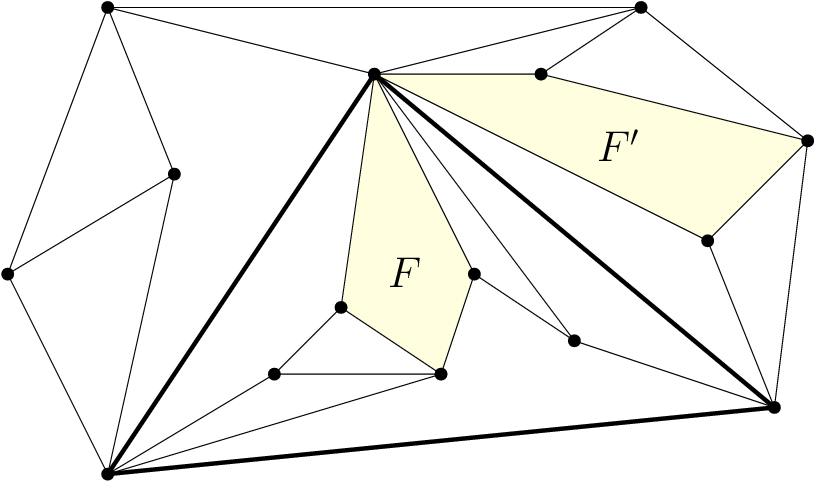}
    \caption{The highlighted triangle is between the faces $F$ and $F'$. It is also between $F$ and the outer face, but it is not between $F'$ and the outer face.
    }
    \label{fig:inside_outside}
\end{figure}

\subsection{Main ideas and preliminary lemmas}

To understand the main idea behind the YES/NO dichotomy for planar instances of \cuts, it helps to imagine that the four terminals form a 4-cycle which bounds the outer face, as is the case for the \Hex~board (recall \autoref{fig:hexboard}). \autoref{fig:exampleplanar} features two such instances.

\begin{figure}[h]
    \centering
    \includegraphics[scale=.64]{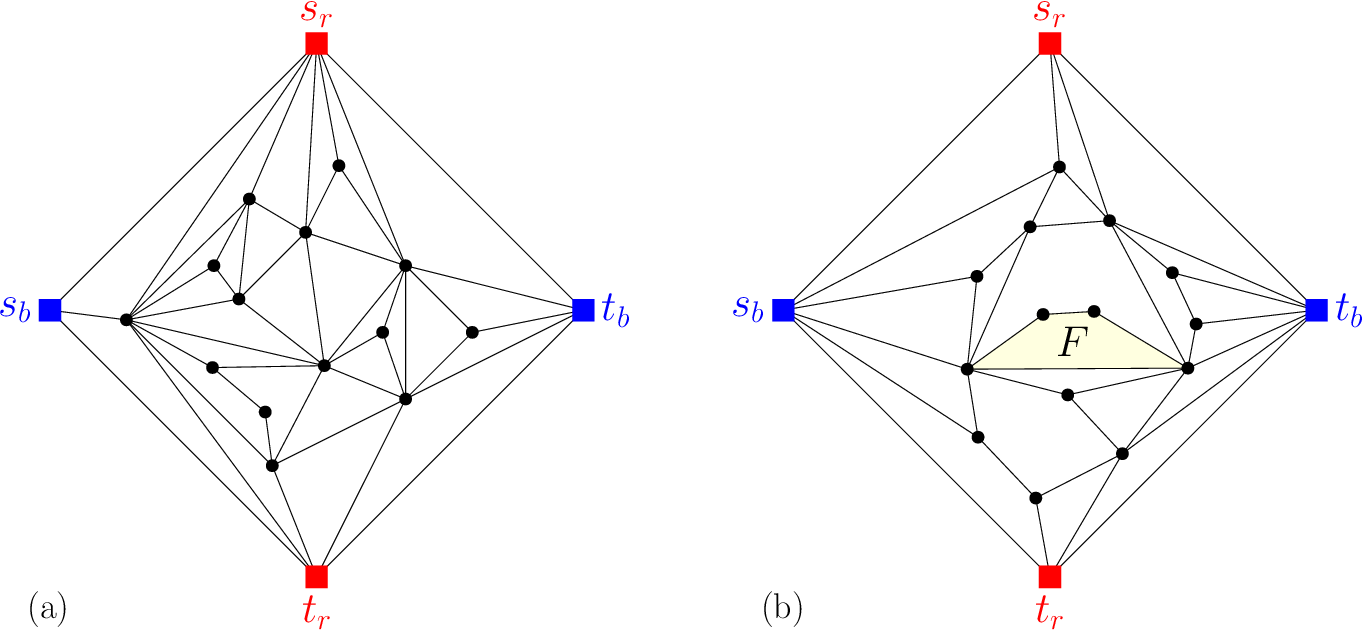}
    \caption{Two planar instances of \cuts: (a) is a NO-instance while (b) is a YES-instance.}
    \label{fig:exampleplanar}
\end{figure}

Gale came up with an elegant algorithm to show that \Hex~cannot end in a draw \cite{Gale1979}, i.e., the corresponding instance of \cuts~is a NO-instance. Schachner noticed that Gale's method could be generalized to any board in which the inner faces are triangles \cite{Schachner2019}. For the sake of self-containment, we now detail the proof of this result, with a slightly stronger statement which only asks that every inner face is inside a triangle. In particular, the following lemma proves that the instance from \autoref{fig:exampleplanar}(a) is a NO-instance.

\begin{lemma}\label{lem:planar-sufficient}
    Let $(G,s_r,t_r,s_b,t_b)$ be an instance of \cuts~such that $G$ is a 2-connected planar graph and $s_r s_b t_r t_b$ is a 4-cycle, which we name $C$. Suppose that there exists a planar embedding of $G$ such that, for every face $F$ inside $C$, there exists a triangle between $F$ and $C$. Then, $(G,s_r,t_r,s_b,t_b)$ is a NO-instance.
\end{lemma}

\begin{proof}
    We assume that $C$ is the outer face and that all inner faces are triangles. Indeed, by our assumption on $G$, we could reduce to that case by removing all vertices strictly outside $C$ and all vertices strictly inside remaining triangles, and then conclude using \autoref{prop:subgraph}.

    Fix a red-blue coloring of the non-terminal vertices. We want to show that there exists a red $s_rt_r$-path or a blue $s_bt_b$-path. For this, we use Gale's algorithm \cite{Gale1979}, which builds a walk in the dual graph of $G$ (\autoref{fig:algorithm} provides an illustration, on the instance from \autoref{fig:exampleplanar}(a) where we have removed the two vertices that were strictly inside a triangle). An edge of $G$ is called {\em mixed} if its endpoints have different colors. We define two finite sequences $(r_i)_i$ and $(b_i)_i$ of red and blue vertices respectively, as follows. We set $r_0=s_r$ and $b_0=s_b$. We start on the outer face, and we ``enter'' the graph through the mixed edge $s_rs_b=r_0b_0$. From there, any time we enter an inner face $r_ib_iv$ through the mixed edge $r_ib_i$, we exit it through its only other mixed edge. If that edge is $r_iv$, i.e., $v$ is blue, then we define $r_{i+1}=r_i$ and $b_{i+1}=v$. Otherwise, we define $r_{i+1}=v$ and $b_{i+1}=b_i$. An inner face that has been visited once will never be visited again, since both its mixed edges have been used. Therefore, at some point, we enter the outer face again through some mixed edge $r_kb_k$: when this happens, the algorithm stops.

\begin{figure}[h]
    \centering
    \includegraphics[scale=.64]{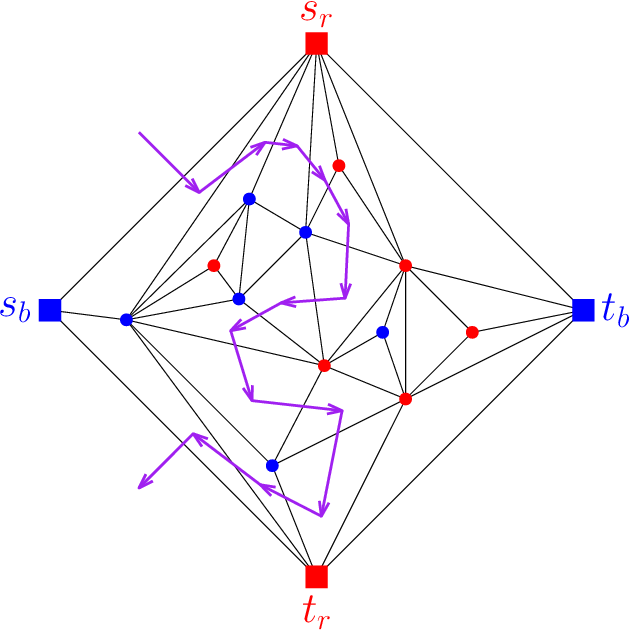}
    \caption{Illustration of Gale's algorithm. In this example, we get the highlighted red $s_rt_r$-path.}
    \label{fig:algorithm}
\end{figure}

We know that $(r_k,b_k) \in \{(s_r,t_b),(t_r,s_b),(t_r,t_b)\}$. As a side note, it is actually impossible that $(r_k,b_k)=(t_r,t_b)$, since we always go through a mixed edge with its red vertex to our left and its blue vertex to our right. In all cases, we have $r_k=t_r$ or $b_k=t_b$. If $r_k=t_r$ then, since $r_{i+1}$ either equals $r_i$ or is a neighbor of $r_i$ by construction, the sequence $(s_r=r_0,r_1\ldots,r_k=t_r)$ contains a red $s_rt_r$-path. If $b_k=t_b$ then, similarly, the sequence $(s_b=b_0,b_1\ldots,b_k=t_b)$ contains a blue $s_bt_b$-path.
\end{proof}

Now, consider the instance from \autoref{fig:exampleplanar}(b). There are faces, such as the highlighted face $F$, which are not inside any triangle, so \autoref{lem:planar-sufficient} does not apply. Actually, we can show that we have a YES-instance, using the following general idea.

By planarity, the vertex set of any $s_rt_r$-path is an $(s_b,t_b)$-separator, and vice-versa. Therefore, the existence of an $s_rt_r$-path and an $s_bt_b$-path which are vertex-disjoint would be a sufficient condition for being a YES-instance. However, such two paths cannot coexist in a planar graph. Instead, we consider ``pseudopaths'', which may follow edges of the graph but are also allowed to traverse faces. The vertex set of any $s_rt_r$-pseudopath is still an $(s_b,t_b)$-separator, and vice-versa. The key difference is that it is actually possible for an $s_rt_r$-pseudopath and an $s_bt_b$-pseudopath to be vertex-disjoint, since they can cross inside a non-triangular face. This is illustrated in \autoref{fig:yes-instance1}, on the example from \autoref{fig:exampleplanar}(b). 

\begin{figure}[h]
    \centering
    \includegraphics[scale=.64]{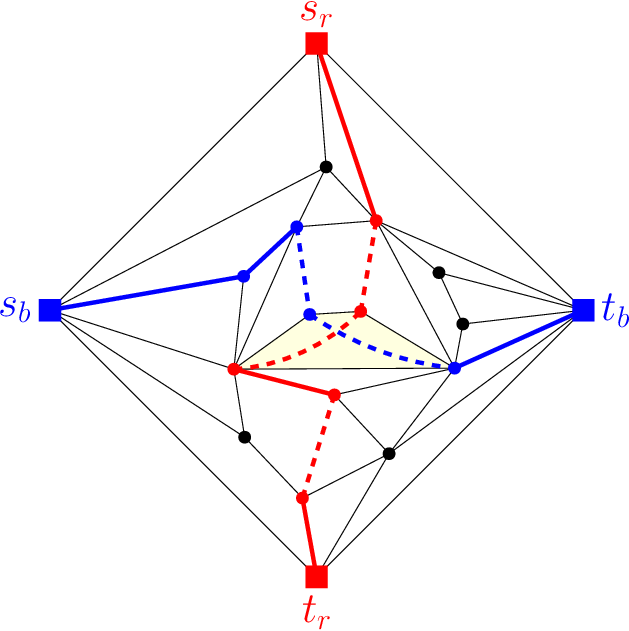}
    \caption{The instance from \autoref{fig:exampleplanar}(b), with two highlighted pseudopaths which cross inside the face $F$, showing that it is a YES-instance.}
    \label{fig:yes-instance1}
\end{figure}

In the rest of this section, we will not use the ``pseudopath'' terminology: instead, it will be convenient to see pseudopaths as actual paths in an augmented graph, which is defined as follows.

\begin{definition}
    Let $G$ be a 2-connected planar graph along with a planar embedding $\phi$. The {\em $\phi$-completion} of $G$ is the graph, denoted by $G^{\phi}$, with same vertex set as $G$ and where two vertices are adjacent in $G^{\phi}$ if and only if there is a face of $G$ that contains both. In other words, $G^{\phi}$ is the graph obtained by ``turning every face of $G$ into a clique'', including the outer face. 
\end{definition}

Note that, if $G$ is a triangulation  (i.e., a planar graph in which every face, including the outer face, is a triangle), then $G^{\phi}=G$. It is well known that all triangulations are 3-connected, and that a triangulation is 4-connected if and only if it has no triangle separator~\cite{Diestel2017,Hakimi1978}. We now generalize this result to all $\phi$-completions of 2-connected planar graphs. The proof uses the following version of Menger's theorem.

\begin{theorem*}[Menger \cite{Menger1927}]
	Let $G=(V,E)$ be a graph and let $s,t \in V$ be distinct non-adjacent vertices. Then, the size of a minimum $(s,t)$-separator in $G$ is equal to the maximum number of internally-vertex-disjoint $st$-paths in $G$.
\end{theorem*}

\begin{lemma}\label{lem:completion}
    Let $G=(V,E)$ be a 2-connected planar graph along with a planar embedding $\phi$. Then:
    \begin{itemize}[nolistsep,noitemsep]
        \item $G^{\phi}$ is 3-connected.
        \item For all distinct non-adjacent $s,t \in V$, the following three assertions are equivalent:
        \begin{enumerate}[noitemsep,nolistsep,label={(\roman*)}]
            \item There exist four pairwise internally-vertex-disjoint $st$-paths in $G^{\phi}$.
            \item There is no triangle $T$ in $G$ such that, for $\phi$, one of $s$ or $t$ is strictly inside $T$ and the other is strictly outside $T$.
            \item There is no triangle $(s,t)$-separator in $G$.
        \end{enumerate}
    \end{itemize}
\end{lemma}

\begin{proof}

    The 3-connectivity of $G^{\phi}$ is straightforward, since it clearly contains a triangulation of $G$ as a spanning subgraph, and all triangulations are 3-connected~\cite{Diestel2017,Hakimi1978}. However, the second part of the lemma cannot be obtained as easily, as triangulating $G$ may create a triangle $(s,t)$-separator, preventing us from using the aforementioned 4-connectivity result for triangulations (actually, it can be easily seen that the equivalence would be false if $G^{\phi}$ was replaced by a triangulation of $G$).

    Let $s,t \in V$ be distinct. It is clear that {$(iii)$}{$\implies$}{$(ii)$}. Moreover, we also have {$(i)$}{$\implies$}{$(iii)$}: indeed, a triangle $(s,t)$-separator in $G$ would also be an $(s,t)$-separator in $G^{\phi}$, so by Menger's theorem it would imply that there do not exist four pairwise internally-vertex-disjoint $st$-paths in $G^{\phi}$. Therefore, it only remains to show that {$(ii)$}{$\implies$}{$(i)$}.

    Suppose that there is no triangle $T$ in $G$ such that, for $\phi$, one of $s$ or $t$ is strictly inside $T$ and the other is strictly outside $T$. We proceed by contradiction and assume that there do not exist four pairwise internally-vertex-disjoint $st$-paths in $G^{\phi}$. By Menger's theorem, this means there exists an $(s,t)$-separator $\{x_1,x_2,x_3\}$ in $G^{\phi}$. Also by Menger's theorem, since $G$ is 2-connected, there exist two internally-vertex-disjoint $st$-paths $P_1=u_0u_1\ldots u_k$ and $P_2=v_0v_1\ldots v_{\ell}$ in $G$, where $u_0=v_0=s$ and $u_k=v_{\ell}=t$. If $P_1$ or $P_2$ contains neither $x_1$, $x_2$ nor $x_3$, then that path is an $st$-path in $G$ (and in $G^{\phi}$) that avoids the separator $\{x_1,x_2,x_3\}$, so we already have our contradiction. Therefore, up to some re-indexing, assume that $P_1$ contains $x_1$ but not $x_2$ nor $x_3$, and that $P_2$ contains at least one of $x_2$ or $x_3$. Define $i \in \{1,\ldots,k-1\}$ as the unique index such that $x_1=u_i$. 

    The key observation is that, for every $x \in V$, $G^{\phi}[N_G[x]]$ contains a spanning wheel centered at $x$ (see $x_1$ in \autoref{fig:pseudopath1} for example). Indeed, if $y_0,\ldots,y_{|N_G(x)|-1}$ is a clockwise list of the neighbors of $x$ in $G$, then $y_j$ and $y_{j+1}$ (where the indices are taken modulo $|N_G(x)|$) sit on a common face of $G$, so they are adjacent in $G^{\phi}$. This wheel can then be used to get around $x$ if we want to build a path that avoids $x$. We will use this method to build an $st$-path in $G^{\phi}$ that avoids $x_1$, $x_2$ and $x_3$, hence the contradiction.

    Let $C$ be the cycle formed by the paths $P_1$ and $P_2$. Let $w_0,w_1,\ldots,w_p$ be the neighbors of $x_1=u_i$ in $G$ whose incident edges to $x_1$ lie inside $C$, listed in the order encountered along $\phi$ from $w_0=u_{i-1}$ to $w_p=u_{i+1}$. Similarly, let $w'_0,w'_1,\ldots,w'_{p'}$ be the neighbors of $x_1=u_i$ in $G$ whose incident edges to $x_1$ lie outside $C$, listed in the order encountered along $\phi$ from $w'_0=u_{i-1}$ to $w'_{p'}=u_{i+1}$. See \autoref{fig:pseudopath1}. 

    \begin{figure}[h]
    \centering
    \begin{tikzpicture}[scale=1.5]

    \draw (0,0) node[vertex] (s){} node[left = .1] {$u_0=s$};
    \draw (0.5,0.7) node[vertex] (v1){} node[left = .1] {};
    \draw (1.2,1.1) node[vertex] (v2){} node[left = .1] {};
    \draw (2,1.4) node[vertex] (v3){} node[left = .1] {};
    \draw (3,1.5) node[vertex] (v4){} node[left = .1] {};
    \draw (4,1.4) node[vertex] (v5){} node[left = .1] {};
    \draw (4.8,1.1) node[vertex] (v6){} node[left = .1] {};
    \draw (5.5,0.7) node[vertex] (v7){} node[left = .1] {};
    \draw (6,0) node[vertex] (t){} node[right = .1] {$t=u_k$};
    \draw (0.8,-0.8) node[vertex] (u1){} node[left = .1] {};
    \draw (1.8,-1.4) node[vertex] (u2){} node[left  = .1] {$u_{i-1}$};
    \draw (3,-1.5) node[vertex] (u3){} node[below = .2] {$x_1$};
    \node[rectangle, line width=1.5, draw] at (3,-1.5) (rect1) {};
    \draw (4.2,-1.4) node[vertex] (u4){} node[right  = .1] {$u_{i+1}$};
    \draw (5.2,-0.8) node[vertex] (u5){} node[left = .1] {};

    \draw (s) -- (u1) -- (u2) -- (u3) -- (u4) -- (u5) -- (t) -- (v7) -- (v6) -- (v5) -- (v4) -- (v3) -- (v2) -- (v1) -- (s);

    \draw (2,-1) node[vertex] (w1){} node[above left] {$w_1$};
    \draw (2.45,-0.8) node[vertex] (w2){} node[above left] {$w_2$};
    \draw (3,-0.7) node[vertex] (w3){} node[left = .1] {};
    \draw (3.55,-0.8) node[vertex] (w4){} node[left = .1] {};
    \draw (4,-1) node[vertex] (w5){} node[above right] {$w_{p-1}$};
    \draw (2,-2) node[vertex] (w'1){} node[below left] {$w'_1$};
    \draw (2.45,-2.2) node[vertex] (w'2){} node[below left] {$w'_2$};
    \draw (3.55,-2.2) node[vertex] (w'4){} node[left = .1] {};
    \draw (4,-2) node[vertex] (w'5){} node[below right] {$w'_{p'-1}$};

    \path (u3) edge (w1) edge (w2) edge (w3) edge (w4) edge (w5) edge (w'1) edge (w'2) edge (w'4) edge (w'5);

    \draw[dashed] (u2) -- (w1) -- (w2) -- (w3) -- (w4) -- (w5) -- (u4) -- (w'5) -- (w'4) -- (w'2) -- (w'1) -- (u2);
    
    \end{tikzpicture}
    \caption{The bottom path is $P_1$ and the top path is $P_2$. Solid lines represent edges of $G$ and dashed lines represent edges of $G^{\phi}$ (which may or may not be edges of $G$). Note that each $w_j$ or $w'_j$ may actually be on the cycle $C$.}
    \label{fig:pseudopath1}
\end{figure}
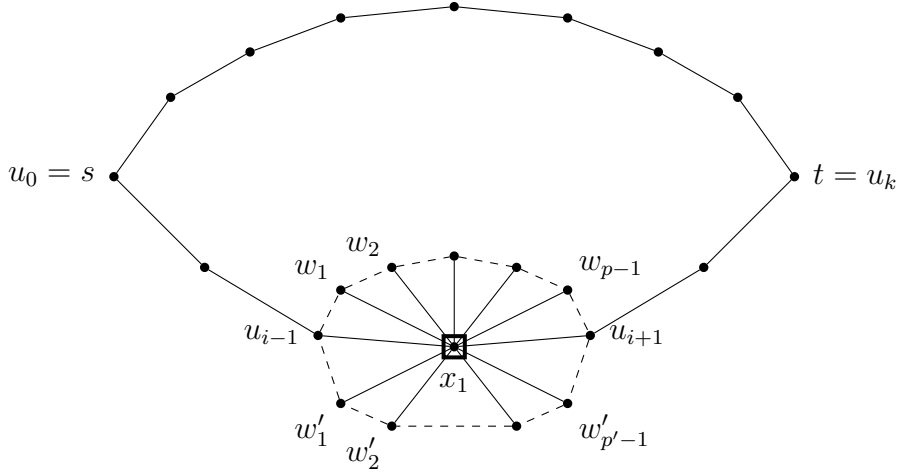

We already know that $\{u_0,\ldots,u_{i-1},u_{i+1},\ldots,u_k\} \cap \{x_2,x_3\} = \varnothing$. If $\{w_1,\ldots,w_{p-1}\} \cap \{x_2,x_3\}=\varnothing$, then the sequence $(u_0,u_1,\ldots,u_{i-1},w_1,\ldots,w_{p-1},u_{i+1},\ldots,u_k)$, which is an $st$-walk in $G^{\phi}$, avoids $x_1$, $x_2$ and $x_3$: this is a contradiction. Similarly, if $\{w'_1,\ldots,w'_{p'-1}\} \cap \{x_2,x_3\}=\varnothing$, then the $st$-walk $(u_0,u_1,\ldots,u_{i-1},w'_1,\ldots,w'_{p'-1},u_{i+1},\ldots,u_k)$ avoids $x_1$, $x_2$ and $x_3$: since this $st$-walk contains an $st$-path, we get a a contradiction.

Therefore, we now suppose that $\{w_1,\ldots,w_{p-1}\} \cap \{x_2,x_3\} \neq \varnothing$ and $\{w'_1,\ldots,w'_{p'-1}\} \cap \{x_2,x_3\} \neq \varnothing$. In particular, this means $x_1$ is a neighbor of both $x_2$ and $x_3$. Without loss of generality, we assume that $x_3\in\{w_1,\ldots,w_{p-1}\}$ and $x_2\in\{w'_1,\ldots,w'_{p'-1}\}$. In particular, it is impossible that $x_2$ and $x_3$ are adjacent in $G$, because $x_1x_2x_3$ would then be a triangle in $G$ such that $s$ or $t$ is strictly inside that triangle and the other is strictly outside, contradicting our assumption on $s$ and $t$. Consider \autoref{fig:pseudopath3} for visual aid (in this figure, $x_2$ and $x_3$ are both on the path $P_2$, but the argument works in all cases): we can see in this figure that, if there was an edge $x_2x_3$ in $G$, then the triangle $x_1x_2x_3$ would enclose $s$ and separate it from $t$. We end the proof by considering three exhaustive cases.

\begin{enumerate}[label={\arabic*)}]
    \item Case 1: the path $P_2$ contains $x_2$ but not $x_3$.
    
    Define $j \in \{1,\ldots,\ell-1\}$ as the unique index such that $x_2=v_j$. Let $y_0,y_1,\ldots,y_q$ be the neighbors of $x_2$ in $G$ whose incident edges to $x_2$ lie inside $C$, listed in the order encountered along $\phi$ from $y_0=v_{j-1}$ to $y_q=v_{j+1}$ (see \autoref{fig:pseudopath2}). Note that $x_1,x_3 \not\in \{y_0,\ldots,y_q\}$ since the edge $x_1x_2$ is outside $C$ and there is no edge $x_2x_3$. Therefore, the sequence $(v_0,v_1,\ldots,v_{j-1},y_1,\ldots,y_{q-1},v_{j+1},\ldots,v_{\ell})$, which is an $st$-walk in $G^{\phi}$, avoids $x_1$, $x_2$ and $x_3$: this is a contradiction.
    
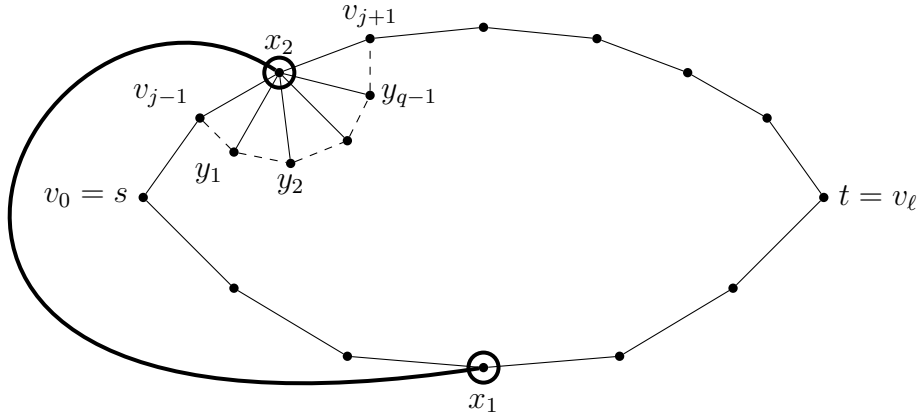
\begin{figure}[h]
    \centering
    \begin{tikzpicture}[scale=1.5,bezier bounding box]

    \draw (0,0) node[vertex] (s){} node[left = .05] {$v_0=s$};
    \draw (0.5,0.7) node[vertex] (v1){} node[above left] {$v_{j-1}$};
    \draw (1.2,1.1) node[vertex] (v2){} node[above = .1] {$x_2$};
    \node[circle, line width=1.5, draw] at (1.2,1.1) (rect1) {};
    \draw (2,1.4) node[vertex] (v3){} node[above] {$v_{j+1}$};
    \draw (3,1.5) node[vertex] (v4){} node[left = .1] {};
    \draw (4,1.4) node[vertex] (v5){} node[below] {};
    \draw (4.8,1.1) node[vertex] (v6){} node[below left = .1] {};
    \draw (5.5,0.7) node[vertex] (v7){} node[below left] {};;
    \draw (6,0) node[vertex] (t){} node[right = .05] {$t=v_{\ell}$};
    \draw (0.8,-0.8) node[vertex] (u1){} node[left = .1] {};
    \draw (1.8,-1.4) node[vertex] (u2){} node[left  = .1] {};
    \draw (3,-1.5) node[vertex] (u3){} node[below = .2] {$x_1$};
    \node[circle, line width=1.5, draw] at (3,-1.5) (rect1) {};
    \draw (4.2,-1.4) node[vertex] (u4){} node[right  = .1] {};
    \draw (5.2,-0.8) node[vertex] (u5){} node[left = .1] {};

    \draw (s) -- (u1) -- (u2) -- (u3) -- (u4) -- (u5) -- (t) -- (v7) -- (v6) -- (v5) -- (v4) -- (v3) -- (v2) -- (v1) -- (s);

    \draw (0.8,0.4) node[vertex] (y1){} node[below left] {$y_1$};
    \draw (1.3,0.3) node[vertex] (y2){} node[below] {$y_2$};
    \draw (1.8,0.5) node[vertex] (y3){} node[below left] {};
    \draw (2,0.9) node[vertex] (y4){} node[right] {$y_{q-1}$};

    \path (v2) edge (y1) edge (y2) edge (y3) edge (y4);

    \draw[dashed] (v1) -- (y1) -- (y2) -- (y3) -- (y4) -- (v3);

    \draw[line width=1.5] (u3)..controls (-3.5,-2.5) and (-1,2.5)..(v2);
    
    \end{tikzpicture}
    \caption{The bottom path is $P_1$ and the top path is $P_2$. Solid lines/curves represent edges of $G$ and dashed lines represent edges of $G^{\phi}$ (which may or may not be edges of $G$).}
    \label{fig:pseudopath2}
\end{figure}

    \item Case 2: the path $P_2$ contains $x_3$ but not $x_2$.
    
    This is analogous to the previous case. Define $j' \in \{1,\ldots,\ell-1\}$ as the unique index such that $x_3=v_{j'}$. Let $y'_0,y'_1,\ldots,y'_{q'}$ be the neighbors of $x_3$ in $G$ whose incident edges to $x_3$ lie outside $C$, listed in the order encountered along $\phi$ from $y'_0=v_{j'-1}$ to $y'_{q'}=v_{j'+1}$. Note that $x_1,x_2 \not\in \{y'_0,\ldots,y'_{q'}\}$ since the edge $x_1x_3$ is inside $C$ and there is no edge $x_2x_3$. Therefore, the sequence $(v_0,v_1,\ldots,v_{j'-1},y'_1,\ldots,y'_{q'-1},v_{j'+1},\ldots,v_{\ell})$, which is an $st$-walk in $G^{\phi}$, avoids $x_1$, $x_2$ and $x_3$: this is a contradiction.
        
    \item Case 3: the path $P_2$ contains both $x_2$ and $x_3$.
    
    Define $j,j' \in \{1,\ldots,\ell-1\}$ as the unique indices such that $x_2=v_j$ and $x_3=v_{j'}$. We define $y_0,y_1,\ldots,y_q$ as in Case 1 and $y'_0,y'_1,\ldots,y'_{q'}$ as in Case 2 (see \autoref{fig:pseudopath3}). We have $x_1,x_3 \not\in \{y_0,\ldots,y_q\}$ and $x_1,x_2 \not\in \{y'_0,\ldots,y'_{q'}\}$. Therefore, if $j<j'$, then the $st$-walk $(v_0,v_1,\ldots,v_{j-1},y_1,\ldots,y_{q-1},v_{j+1},v_{j+2},\ldots,v_{j'-1},y'_1,\ldots,y'_{q'-1},v_{j'+1},\ldots,v_{\ell})$ in $G^{\phi}$ avoids $x_1$, $x_2$ and $x_3$, a contradiction. Symmetrically, if $j'<j$, then the $st$-walk $(v_0,v_1,\ldots,v_{j'-1},y'_1,\ldots,y'_{q'-1},v_{j'+1},v_{j'+2},\ldots,v_{j-1},y_1,\ldots,y_{q-1},v_{j+1},\ldots,v_{\ell})$ in $G^{\phi}$ avoids $x_1$, $x_2$ and $x_3$, a contradiction. 
\end{enumerate}

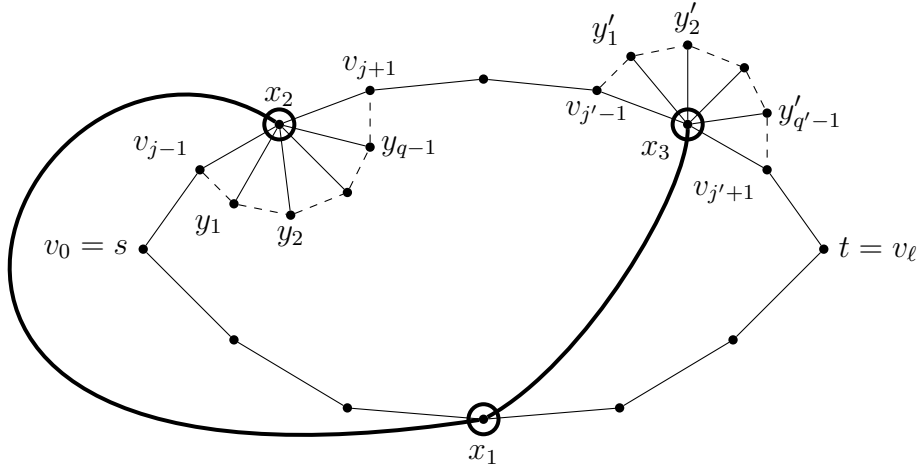
\begin{figure}[h]
    \centering
    \begin{tikzpicture}[scale=1.5,bezier bounding box]

    \draw (0,0) node[vertex] (s){} node[left = .05] {$v_0=s$};
    \draw (0.5,0.7) node[vertex] (v1){} node[above left] {$v_{j-1}$};
    \draw (1.2,1.1) node[vertex] (v2){} node[above = .1] {$x_2$};
    \node[circle, line width=1.5, draw] at (1.2,1.1) (rect1) {};
    \draw (2,1.4) node[vertex] (v3){} node[above] {$v_{j+1}$};
    \draw (3,1.5) node[vertex] (v4){} node[left = .1] {};
    \draw (4,1.4) node[vertex] (v5){} node[below] {$v_{j'-1}$};
    \draw (4.8,1.1) node[vertex] (v6){} node[below left = .1] {$x_3$};
    \node[circle, line width=1.5, draw] at (4.8,1.1) (rect1) {};
    \draw (5.5,0.7) node[vertex] (v7){} node[below left] {$v_{j'+1}$};;
    \draw (6,0) node[vertex] (t){} node[right = .05] {$t=v_{\ell}$};
    \draw (0.8,-0.8) node[vertex] (u1){} node[left = .1] {};
    \draw (1.8,-1.4) node[vertex] (u2){} node[left  = .1] {};
    \draw (3,-1.5) node[vertex] (u3){} node[below = .2] {$x_1$};
    \node[circle, line width=1.5, draw] at (3,-1.5) (rect1) {};
    \draw (4.2,-1.4) node[vertex] (u4){} node[right  = .1] {};
    \draw (5.2,-0.8) node[vertex] (u5){} node[left = .1] {};

    \draw (s) -- (u1) -- (u2) -- (u3) -- (u4) -- (u5) -- (t) -- (v7) -- (v6) -- (v5) -- (v4) -- (v3) -- (v2) -- (v1) -- (s);

    \draw (0.8,0.4) node[vertex] (y1){} node[below left] {$y_1$};
    \draw (1.3,0.3) node[vertex] (y2){} node[below] {$y_2$};
    \draw (1.8,0.5) node[vertex] (y3){} node[below left] {};
    \draw (2,0.9) node[vertex] (y4){} node[right] {$y_{q-1}$};

    \path (v2) edge (y1) edge (y2) edge (y3) edge (y4);

    \draw[dashed] (v1) -- (y1) -- (y2) -- (y3) -- (y4) -- (v3);

    \draw (4.3,1.7) node[vertex] (y'1){} node[above left] {$y'_1$};
    \draw (4.8,1.8) node[vertex] (y'2){} node[above] {$y'_2$};
    \draw (5.3,1.6) node[vertex] (y'3){} node[below left] {};
    \draw (5.5,1.2) node[vertex] (y'4){} node[right] {$y'_{q'-1}$};

    \path (v6) edge (y'1) edge (y'2) edge (y'3) edge (y'4);

    \draw[dashed] (v5) -- (y'1) -- (y'2) -- (y'3) -- (y'4) -- (v7);

    \draw[line width=1.5] (u3)..controls (3.9,-1) and (4.8,0.5)..(v6);
    \draw[line width=1.5] (u3)..controls (-3.5,-2.5) and (-1,2.5)..(v2);
    
    \end{tikzpicture}
    \caption{The bottom path is $P_1$ and the top path is $P_2$. Solid lines/curves represent edges of $G$ and dashed lines represent edges of $G^{\phi}$ (which may or may not be edges of $G$). In this example, we have $j<j'$.}
    \label{fig:pseudopath3}
\end{figure}

In conclusion, we have reached a contradiction in all cases.
\end{proof}

\subsection{Polynomial-time algorithm in the case where the terminals form a cycle}

We provide a structural characterization of YES-instances in the 2-connected planar case when the four terminals form a cycle.

\begin{theorem}\label{thm:planar}
    Let $(G,s_r,t_r,s_b,t_b)$ be an instance of \cuts~such that $G$ is a 2-connected planar graph and $s_rs_bt_rt_b$ is a 4-cycle, which we name $C$. For a given planar embedding $\phi$ of $G$, let $\cP(\phi)$ be the following property: ``there exist faces $F_1$ and $F_2$, the former inside $C$ and the latter outside $C$, such that there is no triangle between $F_1$ and $C$ and no triangle between $F_2$ and $C$''. The following three assertions are equivalent:
        \begin{enumerate}[noitemsep,nolistsep,label={(\arabic*)}]
            \item $(G,s_r,t_r,s_b,t_b)$ is a YES-instance.
            \item There exists a planar embedding $\phi$ of $G$ such that property $\cP(\phi)$ holds.
            \item For every planar embedding $\phi$ of $G$, property $\cP(\phi)$ holds.
        \end{enumerate}
\end{theorem}

Note that the faces $F_1$ and $F_2$ cannot be triangles: for instance, if $F_1$ was a triangle, then it would be inside itself and $C$. Also note that $C$ is a possible candidate for the face $F_2$ if $C$ is the outer cycle: the YES-instance from \autoref{fig:exampleplanar}(b) illustrates this. See \autoref{fig:yes-instance2} for another example of a YES-instance.

\begin{figure}[h]
    \centering
    \includegraphics[width=\linewidth]{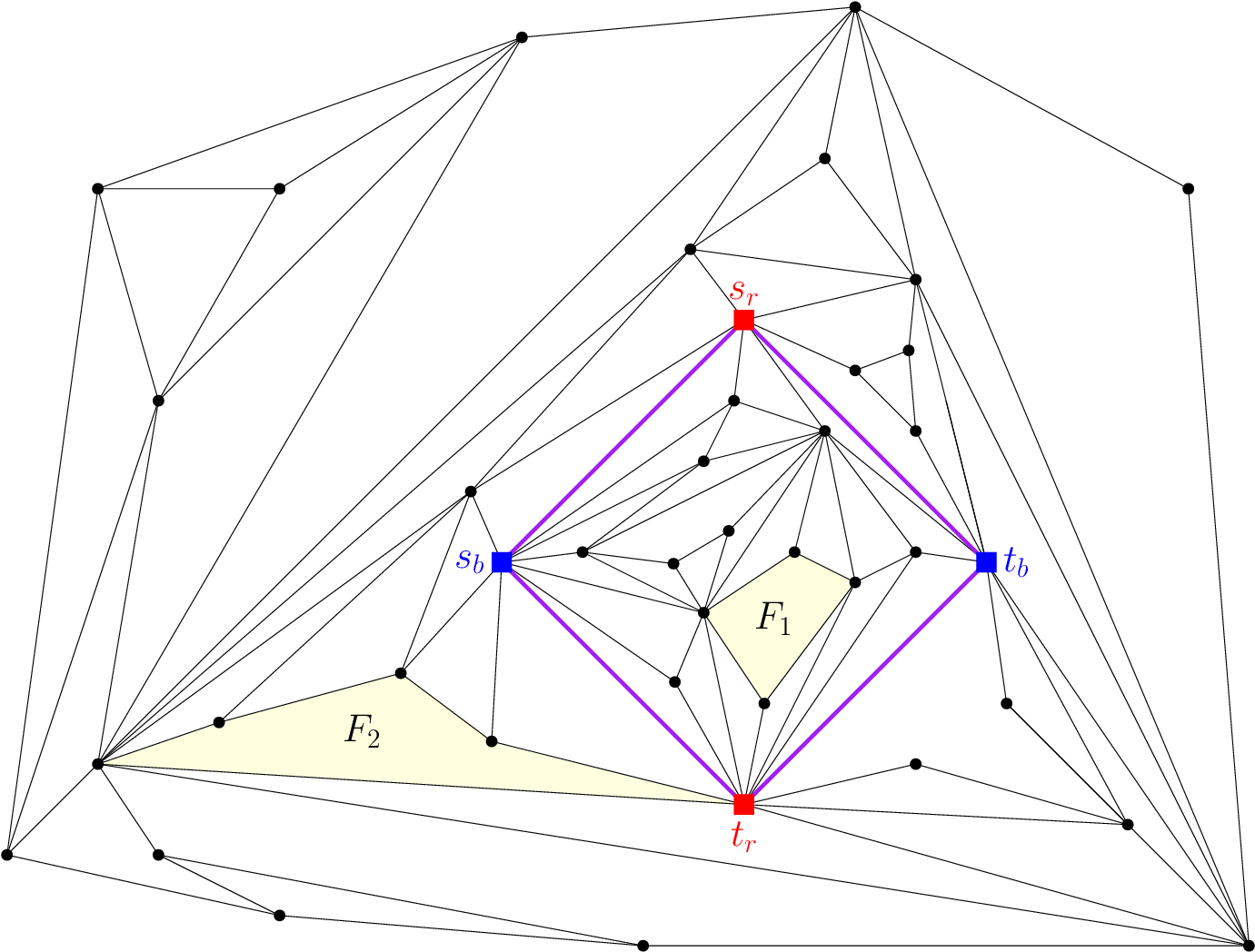}
    \caption{An example of a planar YES-instance where the four terminals form a cycle. The highlighted faces $F_1$ and $F_2$ satisfy the property from \autoref{thm:planar} (actually, it can be checked that they are the only ones that do).
    }
    \label{fig:yes-instance2}
\end{figure}

\begin{proof}
    Let $(G,s_r,t_r,s_b,t_b)$ be an instance of \cuts~such that $G=(V,E)$ is a planar 2-connected graph and $s_rs_bt_rt_b$ is a 4-cycle $C$. 
    It is obvious that {\textit{(3)}}{$\implies$}{\textit{(2)}}. Next, let us address {\textit{(1)}}{$\implies$}{\textit{(3)}}, by showing the contrapositive: suppose that there exists either a planar embedding $\phi_{in}$ of $G$ such that there is a triangle between $F$ and $C$ for every face $F$ inside $C$, or a planar embedding $\phi_{out}$ of $G$ such that there is a triangle between $F$ and $C$ for every face $F$ outside $C$. We note that the existence of $\phi_{in}$ and $\phi_{out}$ are actually equivalent, since any planar embedding admits a ``mirror'' planar embedding with respect to any given cycle, where the inside and the outside of that cycle are swapped. Therefore, we assume that $\phi_{in}$ exists. \autoref{lem:planar-sufficient} immediately concludes that $(G,s_r,t_r,s_b,t_b)$ is a NO-instance, so {\textit{(1)}}{$\implies$}{\textit{(3)}}. Finally, we show that {\textit{(2)}}{$\implies$}{\textit{(1)}}. 
    Let $\phi$ be a planar embedding of $G$ such that property $\cP(\phi)$ holds.
    
     Let $V_{in}$ (resp. $V_{out})$ denote the set of vertices of $G$ that are inside $C$ (resp. outside $C$) for $\phi$: we have $V_{in} \cup V_{out} = V$ and $V_{in} \cap V_{out} = C$. Note that $G[V_{in}]$ and $G[V_{out}]$ are also planar and 2-connected. Moreover, by planarity, $(G,s_r,t_r,s_b,t_b)$ is a YES-instance if and only if both $(G[V_{in}],s_r,t_r,s_b,t_b)$ and $(G[V_{out}],s_r,t_r,s_b,t_b)$ are YES-instances. 
    Let $\phi'$ be the mirror image of $\phi$ with respect to $C$. Since $\cP(\phi')$ also holds, it suffices to show that $(G[V_{in}],s_r,t_r,s_b,t_b)$ is a YES-instance.
    
    Let $F$ be a face of $G$ inside $C$ such that there is no triangle between $F$ and $C$. Recall that, in particular, this means $F$ is not a triangle. 
    Let $G_0$ be the graph obtained from $G[V_{in}]$ by:
    \begin{itemize}[nolistsep,noitemsep]
        \item[--] adding a new vertex $s$ strictly inside $F$ and edges from $s$ to all vertices of $F$;
        \item[--] adding a new vertex $t$ strictly outside $C$ and edges from $t$ to all four vertices of $C$.
    \end{itemize}

    Note that $G_0$ is still a planar 2-connected graph. Let $\phi_0$ be a planar embedding of $G_0$ that coincides with $\phi$ on the subgraph $G[V_{in}]$. \autoref{fig:G0phi} provides an illustration of the graphs $G_0$ and $G_0^{\phi_0}$.

    \begin{figure}[h]
    \centering
    \includegraphics[width=\linewidth]{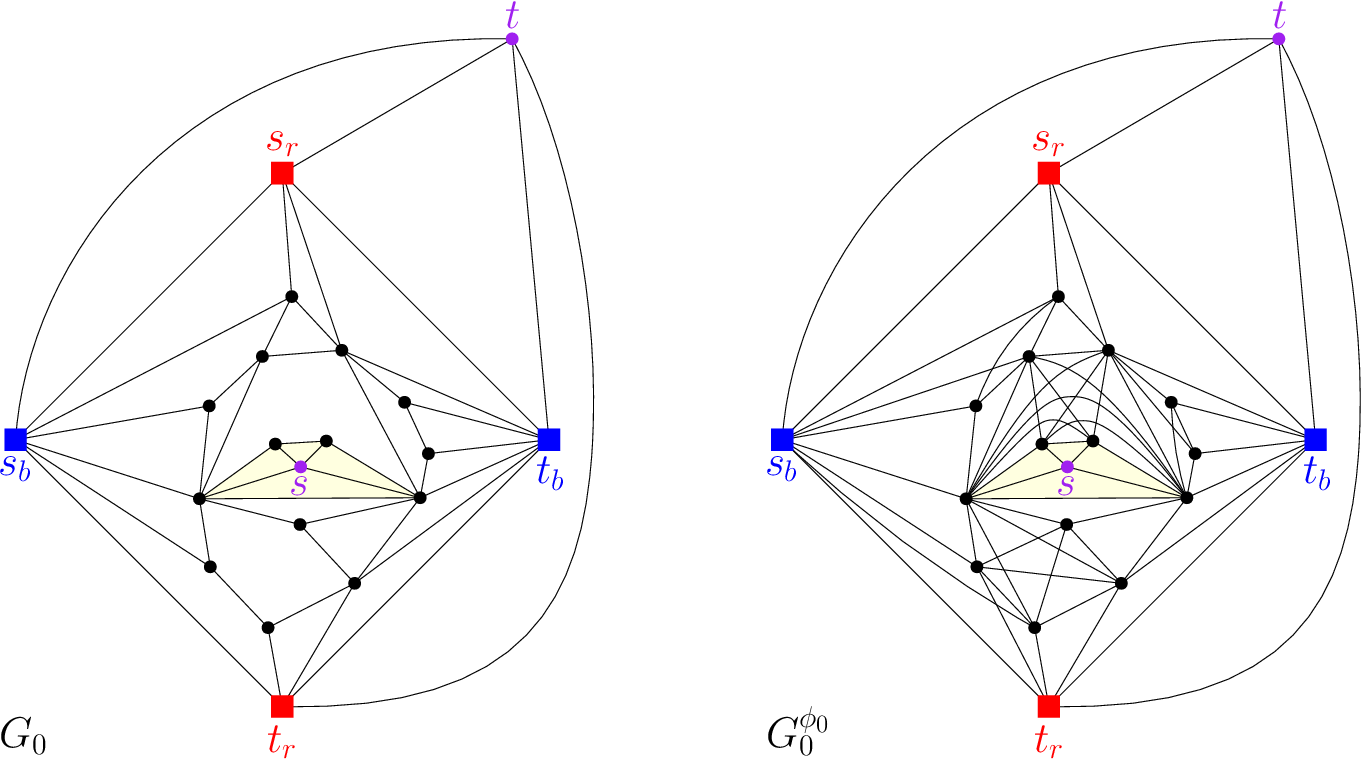}
    \caption{The graphs $G_0$ and $G_0^{\phi_0}$ on the example from \autoref{fig:exampleplanar}(b). The two pseudopaths from \autoref{fig:yes-instance1} can be obtained from four pairwise internally-vertex-disjoint $st$-paths in $G_0^{\phi_0}$.
    }
    \label{fig:G0phi}
    \end{figure}
    
    By our assumption on the face $F$, there is no triangle $T$ in $G_0$ such that one of $s$ or $t$ is strictly inside $T$ and the other is strictly outside $T$. Therefore, by \autoref{lem:completion}, there exist four pairwise internally-vertex-disjoint $st$-paths $P_1,P_2,P_3,P_4$ in $G_0^{\phi_0}$. Note that the edges incident to $t$ in $G_0^{\phi_0}$ are the same as in $G_0$, since $t$ was already a neighbor of all four vertices of $C$ in $G_0$. Without loss of generality, assume that the neighbor of $t$ in $P_1$ (resp. $P_2$, $P_3$, $P_4$) is $s_r$ (resp. $t_r$, $s_b$, $t_b$). 
    Let $P_r$ be the $s_rt_r$-path in $G_0^{\phi_0}$ obtained by taking $(P_1 \cup P_2)-t$ and replacing the two edges incident to $s$ by a single edge going across $F$. Similarly, let $P_b$ be the $s_bt_b$-path in $G_0^{\phi_0}$ obtained by taking $(P_3 \cup P_4)-t$ and replacing the two edges incident to $s$ by a single edge going across $F$. Let $X$ be the vertex set of $P_b$, and let $Y$ be the vertex set of $P_r$. Clearly, $X \cap Y = \varnothing$.
    
    To conclude, we show that $(X,Y)$ is an $(s_r,t_r,s_b,t_b)$-separator in $G[V_{in}]$. 
 Even though this fact is clear visually, let us give a rigorous proof. We only consider $X$, as the case of $Y$ is symmetric. Suppose for a contradiction that $X$ is not an $(s_r,t_r)$-separator in $G$, that is, there exists an $s_rt_r$-path $P$ in $G-X$. Let $G_1$ be the graph obtained from $G[V_{in}]$ by:
    \begin{itemize}[nolistsep,noitemsep]
        \item[--] adding all edges of $P_b$, so that $P_b$ is a path in $G_1$;
        \item[--] adding a new vertex $u$ strictly outside $C$ and edges from $u$ to all four vertices of $C$.
    \end{itemize}
    Note that $G_1$ is still a planar 2-connected graph. We claim that $G_1$ contains a subdivision of $K_5$ rooted on $\{u,s_r,t_r,s_b,t_b\}$. Indeed, $G_1$ possesses the following eight edges: $s_rs_b$, $s_bt_r$, $t_rt_b$, $t_bs_r$, $us_r$, $ut_r$, $us_b$, $ut_b$. Moreover, we have the $s_bt_b$-path $P_b$ in $G_1$, which contains neither $s_r$, $t_r$ nor $u$. Finally, we also have the $s_rt_r$-path $P$ in $G_1$, which is vertex-disjoint from $P_b$ and does not contain $u$ either. See \autoref{fig:kuratowski} for an illustration. This contradicts Kuratowski's theorem \cite{Kuratowski1930}, which states that a graph is planar if and only if it contains no subdivision of $K_5$ or $K_{3,3}$. 
\end{proof}

\begin{figure}[h]
    \centering
    \includegraphics[scale=0.5]{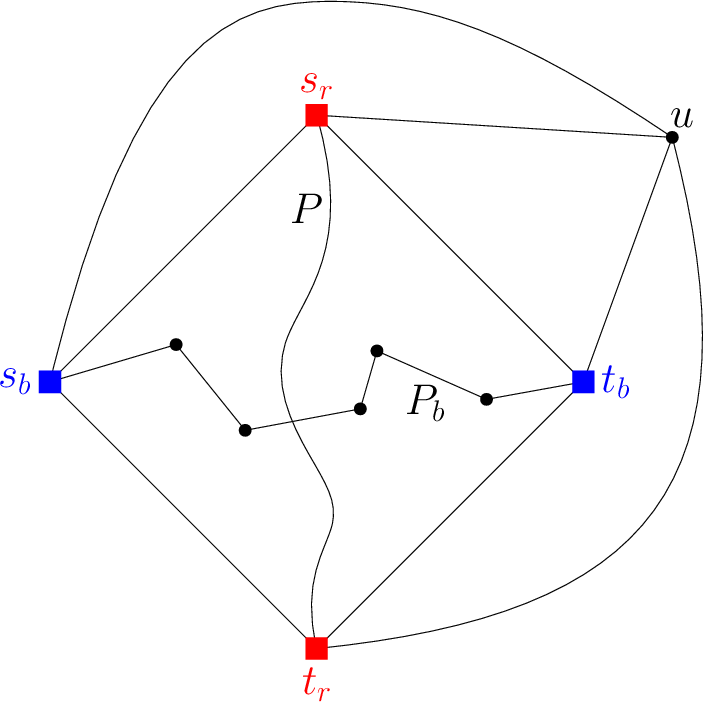}
    \caption{A subdivision of $K_5$ in the planar graph $G_1$, contradicting Kuratowski's theorem.
    }
    \label{fig:kuratowski}
\end{figure}

From the structural characterization provided by \autoref{thm:planar}, we can derive a polynomial-time algorithm.

\begin{corollary}\label{cor:planar}
    There is a quadratic-time algorithm which solves \cuts~on instances $(G,s_r,t_r,s_b,t_b)$ where $G$ is a planar graph and $s_rs_bt_rt_b$ is a 4-cycle.
\end{corollary}

\begin{proof}
    Let $(G,s_r,t_r,s_b,t_b)$ be an instance of \cuts~where $G$ is a planar graph and $s_rs_bt_rt_b$ is a 4-cycle, which we name $C$. We may assume, by \autoref{prop:2-connected}, that $G$ is 2-connected up to a preliminary step in $O(n^2)$ time (indeed, since we reduce to an induced subgraph, planarity is preserved as well as the fact that $s_rs_bt_rt_b$ is a 4-cycle). Next, we fix a straight-line planar embedding $\phi$ of $G$, which can be computed in $O(n)$ time \cite{CHIBA198554}, and we enumerate the $O(n)$ faces and $O(n)$ triangles of $G$ in $O(n)$ time \cite{listtriangles}.
    
    By \autoref{thm:planar}, $(G,s_r,t_r,s_b,t_b)$ is a YES-instance if and only if there exist faces $F_1$ and $F_2$ (which are necessarily not triangles), the former inside $C$ and the latter outside $C$, such that there is no triangle between $F_1$ and $C$ and no triangle between $F_2$ and $C$. To check for the existence of $F_1$, we start by removing all the vertices that are strictly outside $C$, which can be done in $O(n)$ time using the coordinates of the vertices in the plane. For each remaining non-triangular face $F$ and triangle $T$, we proceed to check whether $F$ is inside $T$, which is done in $O(1)$ time for each given $(F,T)$: indeed, if $x$ is any vertex of $F$ that is not a vertex of $T$, then $F$ is inside $T$ if and only if $x$ is strictly inside $T$. In this way, the existence of $F_1$ can be checked in $O(n^2)$ time. The case of $F_2$ is addressed in similar fashion.
    
    We conclude that the algorithm runs in $O(n^2)$ time.
\end{proof}

\subsection{Polynomial-time algorithm in the general planar case}

We are now ready to prove this section's main result by reducing the general planar case of \cuts~to the subcase where the four terminals form a cycle.

\begin{theorem}\label{theo:algoplanar}
	There is a quadratic-time algorithm which solves \cuts~on instances $(G,s_r,t_r,s_b,t_b)$ where $G$ is a planar graph.
\end{theorem}

\begin{proof}
	Let $(G,s_r,t_r,s_b,t_b)$ be an instance of \cuts~where $G$ is a planar graph. We start by explaining how, after a first step in $O(n^2)$ time, the problem is either solved or reduced to a planar instance having the following properties:
	\begin{enumerate}[label={(\arabic*)},noitemsep,nolistsep]
		\item $s_r$ and $t_r$ are not adjacent, and neither are $s_b$ and $t_b$; 
		\item $\dist_G(s_r,s_b) \leq 2$, $\dist_G(s_r,t_b) \leq 2$, $\dist_G(t_r,s_b) \leq 2$ and $\dist_G(t_r,t_b) \leq 2$;
		\item every non-terminal vertex is a neighbor of at most two of the four terminals.
	\end{enumerate}
By definition of the \cuts~problem, property (1) holds at the start. 
We can then check property (2) in $O(n)$ time, and we have a YES-instance by \autoref{prop:distance} if it does not hold. Now, suppose that some vertex $x$ is a neighbor of at least three terminals. If $x$ is a neighbor of all four terminals, then we have a trivial NO-instance. Otherwise, assume without loss of generality that $x$ is a neighbor of $s_r$, $t_r$ and $s_b$, but not $t_b$. Then, coloring $x$ in blue is forced and, since $x$ is a neighbor of $s_b$, $x$ behaves like a new blue source. Therefore, the instance $(G,s_r,t_r,s_b,t_b)$ is equivalent to the instance $(G,s_r,t_r,\{s_b,x\},t_b)$ of \generalcuts. By \autoref{prop:contraction}, we can then contract the edge $s_bx$ and get an equivalent instance of \cuts. Note that edge contractions preserve planarity. Whenever we find a vertex $x$ that is a neighbor of three terminals, we eliminate it in $O(n)$ time by performing an edge contraction, and then we check in $O(1)$ time that property (1) still holds (otherwise, we obviously conclude that we have a NO-instance). There is no need to check for property (2) again since edge contractions cannot increase distances. After some $O(n^2)$ time, we thus reduce to an instance of \cuts~which satisfies properties (1), (2) and (3).

We now start the second and final step of the algorithm, which is illustrated in \autoref{fig:4cycle}. Property (2) ensures that, for each of the four pairs $\{s_r,s_b\}$, $\{s_r,t_b\}$, $\{t_r,s_b\}$ and $\{t_r,t_b\}$, either the two vertices of the pair are adjacent or we can designate an ``intermediate vertex'' which is a neighbor of both. Moreover, by property (3), no two of these four pairs can share the same intermediate vertex. Therefore, the four terminals plus the intermediate vertices form a cycle of length between 4 and 8, where $s_r$, $s_b$, $t_r$ and $t_b$ appear in that order, some of them possibly separated by an intermediate vertex. Now, since there are at most four intermediate vertices, we can use brute force to determine their color and eliminate them in $O(n)$ time via edge contractions similarly to what was done in the first step of the algorithm. For instance, say $s_r$ and $s_b$ are not adjacent, and let $x$ be an intermediate vertex that is a neighbor of both: If we color $x$ in, say, red, then $x$ behaves as a new red source (as it is a red neighbor of the red source $s_r$) and so we may contract the edge $s_rx$. For each color combination of the intermediate vertices, and after the edge contractions, we get an instance in which $s_rs_bt_rt_b$ is a 4-cycle, which we solve by calling the algorithm in $O(n^2)$ time from \autoref{cor:planar}. Thus, the second step is performed in $O(n^2)$ time. 

In conclusion, both steps are achieved in $O(n^2)$ time, leading to the desired quadratic-time algorithm.
\end{proof}

\begin{figure}[h]
    \centering
    \includegraphics[scale=.64]{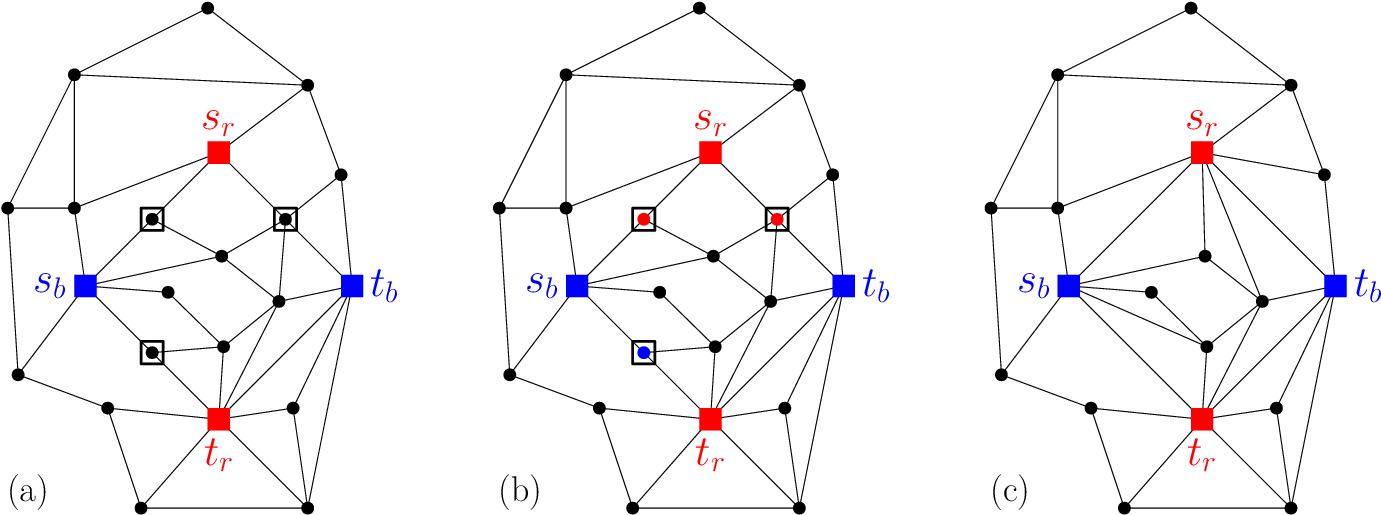}
    \caption{Illustration of the second step. (a) Initial situation, where the vertices squared in black are the intermediate vertices. (b) One of eight brute-force combinations for the colors of the intermediate vertices. (c) Three edge contractions yield the desired 4-cycle.}
    \label{fig:4cycle}
\end{figure}

\begin{corollary}
    There is a quadratic-time algorithm which solves \generalcuts~on instances $(G,S_r,T_r,S_b,T_b)$ where $G$ is a planar graph and each of the four sets $S_r$, $T_r$, $S_b$ and $T_b$ induces a connected subgraph.
\end{corollary}

\begin{proof}
    We can perform $O(n)$ edge contractions to contract each of $S_r$, $T_r$, $S_b$ and $T_b$ into a single terminal, which does not change the nature of the instance by \autoref{prop:contraction}. After this step, which takes some $O(n^2)$ time, we apply \autoref{theo:algoplanar}.
\end{proof}

\section{Conclusion}\label{sec:conclusions}

\subsection{Summary}

In this work, we introduced and studied the disjoint separators problem, motivated by a natural generalization of the \Hex~game to arbitrary graphs. Our main contributions can be summarized as follows.

On the negative side, we prove that \generalcuts~is \NP-complete even under strong structural restrictions, including planar graphs of bounded maximum degree. This indicates that the problem remains computationally difficult even in sparse graphs with additional structural constraints.

On the positive side, we show that \cuts~is polynomial-time solvable on planar graphs. We give a structural characterization of YES-instances when the four terminals form a cycle, and we extend this approach to arbitrary planar instances by reducing them to this case. 

As a corollary, we obtain an almost tight complexity dichotomy: \cuts~is tractable on planar graphs, but it is \NP-hard on instances where removing the four terminals yields a planar graph of bounded maximum degree.

\subsection{Future work}

Several directions naturally arise from this work.

\begin{itemize}[leftmargin=*,noitemsep]
    \item \textbf{Graph classes:}
    A first direction is to extend the study beyond planar graphs. In particular, graphs of bounded genus appear as a natural next step, where topological obstructions may interact non-trivially with the separation constraints.

    \item \textbf{Minimization variants:}
    One may consider optimization versions of the problem, where, instead of deciding existence, one seeks disjoint separators optimizing some objective.  Natural objectives include minimizing the size of the larger of the two separators, minimizing the size difference between the two separators, or minimizing the size of their union. These variants may exhibit different complexity behavior and could be studied under standard structural restrictions such as bounded treewidth or bounded degeneracy.

    \item \textbf{Coinciding terminals:}
    It would be interesting to investigate special cases where terminals coincide or are highly correlated, for instance when $s_r = s_b$, or more generally when the sets of terminals are not pairwise disjoint.

    \item \textbf{Terminals-only variant:}
    A related setting arises when no distinguished sources and targets are specified. In this case, one asks whether a graph admits a red-blue coloring of the vertices such that there is no monochromatic path between any two terminals of the same color. This can be seen as a global avoidance version of the problem.

    \item \textbf{$k$-player generalizations:}
    Our problem admits two natural generalizations to an arbitrary number $k$ of pairs of terminals, with each pair having its own color. The first one is about finding pairwise disjoint separators of all the pairs. The second one consists in coloring the non-terminal vertices so that no monochromatic source-to-target path is created. Note that these two problems coincide for $k=2$, but become different for $k \geq 3$.
    
\end{itemize}

\section*{Acknowledgments}

The authors thank Pierre Guillon and K\'evin Perrot for inspiring the topic of this paper.

\bibliographystyle{abbrv}
\bibliography{references}

\end{document}